\documentclass[11pt]{article}
\pdfoutput=1

\usepackage{amsfonts,amsmath,amsthm,graphicx,nicefrac,amssymb}
\usepackage{fullpage}
\usepackage[margin=1in]{geometry}
\usepackage{natbib}
\usepackage{changepage}
\usepackage{enumerate}
\usepackage{scalerel}
\usepackage{accents}
\usepackage{hyperref}
\usepackage{xfrac}
\usepackage{xifthen}% provides \isempty test

\usepackage{setspace,xspace}
\usepackage[title,titletoc]{appendix}

\usepackage{subfig}
\usepackage{tikz}
\usetikzlibrary{patterns}
\usetikzlibrary{intersections}

\numberwithin{equation}{section}

\usepackage{slivkins-setup,slivkins-theorems-original}
\renewcommand{\xhdr}[1]{\vspace{1.5mm} \noindent{\bf #1}}

\input{widebar}

\newenvironment{proofof}[1][]
       {\begin{proof}[Proof of #1]}
       {\end{proof}}

\newcommand{\loS}{S_{\mathtt{L}}} % rounds $t$ with small p_t(a_0)
\newcommand{\hiS}{S_{\mathtt{H}}} % rounds $t$ with large p_t(a_0)
\newcommand{\hiN}{N_{\mathtt{H}}} % sum_{t\in \hiS} q_t
\newcommand{\term}[1]{\ensuremath{\mathtt{#1}}\xspace}

\newcommand{\support}{\term{support}}

\newcommand{\ips}{\term{IPS}}

\newcommand{\var}{\term{Var}}

\newcommand{\mech}{\term{mech}} % notation for generic mechanism
\newcommand{\adv}{\term{adv}}   % notation for adversary
\newcommand{\mechname}{\term{TwoStageMech}} % name for our mechanism

\newcommand{\val}{u}
\newcommand{\Val}{U}
\newcommand{\ValOut}{U_{\term{out}}}

\newcommand{\bestVal}[1][\prior]{\Val^*_{#1}}

\newcommand{\history}{\term{hist}}

\newcommand{\err}[1][]{\term{ERR}^{#1}}

\newcommand{\benchmark}{\term{bench}}
\newcommand{\bench}{\benchmark}
\newcommand{\worstbench}[1][\prior]{\bench(#1,\types)}
\newcommand{\greedy}[2]{\ensuremath{(#1,#2)\text{-\term{greedy}}}} % greedy policy: \greedy{\eps}{\mystate}

\newcommand{\type}{\theta}
\newcommand{\types}{\Theta}
\newcommand{\subtype}{\type_{\term{pri}}}   % private (subjective) type
\newcommand{\subtypes}{\types_{\term{pri}}} % set of private types
\newcommand{\pubtype}{\type_{\term{pub}}}   % public (objective) type
\newcommand{\pubtypes}{\types_{\term{pub}}} % set of public types
\newcommand{\reptype}{\type_{\term{rep}}}  % reported (private) type
\newcommand{\obstype}{\type_{\term{obs}}}  % observed type

\newcommand{\policy}{\sigma}

\newcommand{\hatpolicy}{\widehat{\policy}}
\newcommand{\benchpolicy}{\policy^{\term{opt}}}
\newcommand{\hatbenchpolicy}{\widehat{\policy}^{\term{opt}}}

\newcommand{\signal}{\term{sig}}

\newcommand{\outside}{\mA_{\term{out}}}

\newcommand{\outcome}{\omega}
\newcommand{\outcomes}{\Omega}

\newcommand{\mystate}{\psi}
\newcommand{\mystates}{\Psi}
\newcommand{\truestate}{\mystate^*}
\newcommand{\approxstate}{\overline{\mystate}}
\newcommand{\estimatestate}{\mystate_\term{MLE}}
\newcommand{\prior}{\mP}
\newcommand{\jointP}{\mQ} % joint prior on (\truestate,\signal)

\newcommand{\initProb}{\delta_{\prior}}
\newcommand{\sampleD}{p}  % sampling distribution chosen by the algorithm in a given round

\newcommand{\bestA}[1][\mystate]{%
\ifthenelse{\isempty{#1}}{a^*}{a^*(#1)}%
}

\newcommand{\agentE}{\mathbb{E}_{\term{Bayes}}} % expectation w.r.t. agents' beliefs
\newcommand{\probE}{\mathbb{P}_{\term{Bayes}}} % expectation w.r.t. agents' beliefs

\newcommand{\truE}[1]{\mE^{\term{tru}}_{#1}} % truthful reporting up to (and including) round t

\newcommand{\advF}[1][]{f_{\term{adv}}^{#1}}            % average scores induced by the adversary
\newcommand{\estF}[1][]{\widehat{f}_{\term{mech}}^{#1}} % mechanism's estimates

\newcommand{\advFreq}[1][]{\advF[#1]} % for compatibility
\newcommand{\estFreq}[1][]{\estF[#1]} % for compatibility

\newcommand{\typeDist}{D_{\term{type}}}
\newcommand{\typeFreq}[1][]{F_{\term{type}}^{#1}}         % type freqs
\newcommand{\typeFreqEst}[1][]{\widetilde{F}_{\term{type}}^{#1}} % estimated type freqs

\newcommand{\gap}{\term{gap}}

\newcommand{\proberr}{p_{\term{err}}}
\newcommand{\llr}{\term{LLR}}
\newcommand{\pmin}{p_{\term{min}}}
\newcommand{\lr}{\mathcal{R}}

\newcommand{\worstbenchhat}[1][\prior]{\widehat{\bench}(#1,\types)}

\newcommand{\KL}{{\rm KL}}

\newcommand{\norm}[2]{\lVert\,#1\,\rVert_{#2}} %\norm{x}{p} = ||x||_p
\renewcommand{\ll}{\ell} %log likelihood ratio
\newcommand{\supp}{{\rm supp}}

\newcommand{\given}{|}

\newcommand{\prob}[2][]{\text{\bf Pr}\ifthenelse{\not\equal{}{#1}}{_{#1}}{}\!\left[{\def\givenn{\middle|}#2}\right]}
\newcommand{\expect}[2][]{\text{\bf E}\ifthenelse{\not\equal{}{#1}}{_{#1}}{}\!\left[{\def\givenn{\middle|}#2}\right]}

\newcommand{\tparen}{\big}
\newcommand{\tprob}[2][]{\text{\bf Pr}\ifthenelse{\not\equal{}{#1}}{_{#1}}{}\tparen[{\def\given{\tparen|}#2}\tparen]}
\newcommand{\texpect}[2][]{\text{\bf E}\ifthenelse{\not\equal{}{#1}}{_{#1}}{}\tparen[{\def\given{\tparen|}#2}\tparen]}

\newcommand{\sprob}[2][]{\text{\bf Pr}\ifthenelse{\not\equal{}{#1}}{_{#1}}{}[#2]}
\newcommand{\sexpect}[2][]{\text{\bf E}\ifthenelse{\not\equal{}{#1}}{_{#1}}{}[#2]}

\newcommand{\abs}[1]{\left|#1\right|}

\usepackage[capitalize]{cleveref}
\crefname{myAppendix}{Appendix}{Appendices}
\Crefname{myAppendix}{Appendix}{Appendices}
\crefname{claim}{Claim}{Claims}
\Crefname{claim}{Claim}{Claims}

\let\oldparagraph\paragraph
\renewcommand{\paragraph}[1]{\oldparagraph{#1.}}

\begin{document}
\title{Exploration and Incentivized Participation in Randomized Trials}

%\title{Exploration and Incentivizing Participation in Clinical Trials}
%\footnote{The authors are grateful to Ian Ball and the journal referees for their insightful comments.}

\author{Yingkai Li\thanks{
National University of Singapore, Singapore.
Email: \texttt{yk.li@nus.edu.sg}. Research carried out while this author was an intern in Microsoft Research NYC and a graduate student at Northwestern University.}
\and
Aleksandrs Slivkins\thanks{Microsoft Research, New York, NY. Email: \texttt{slivkins@microsoft.com }}}

\date{First version: February 2022\\This version: January 2026%
\footnote{\emph{Version history.} The revisions mainly concerned  presentation, discussions, and related work. Previous versions (before Jan'25) focused on clinical RCTs as an application domain, and were titled ``Exploration and Incentivizing Participation in Clinical Trials" and ``Incentivizing Participation in Clinical Trials" (pre-2024).
The Mar'24 revision allowed an arbitrary outside option (as opposed to \refeq{eq:model-outside}).}}

%added a more general interpretation of the ``outside option", see \cref{rem:model-outside}.

\maketitle

\begin{abstract}
%The difficulty of recruiting patients is
Participation incentives is a well-known issue inhibiting randomized controlled trials (RCTs) in medicine, as well as a potential cause of user dissatisfaction for RCTs in online platforms. We frame this issue as a non-standard exploration-exploitation tradeoff: an RCT would like to explore as uniformly as possible, whereas each ``agent" (a patient or a user) prefers ``exploitation'', \ie treatments that seem best. We incentivize participation by leveraging information asymmetry between the trial and the agents. We measure statistical performance via worst-case estimation error under adversarially generated outcomes, a standard objective for RCTs. We obtain a near-optimal solution in terms of this objective: an incentive-compatible mechanism with a particular guarantee, and a nearly matching impossibility result for any incentive-compatible mechanism. We consider three model variants: homogeneous agents (of the same ``type" comprising beliefs and preferences), heterogeneous agents, and an extension that leverages estimated type frequencies to mitigate the influence of rare-but-difficult agent types.

% 2024-05 version focusing on medical trials

%Participation incentives is a well-known issue inhibiting randomized clinical trials (RCTs). We frame this issue as a non-standard exploration-exploitation tradeoff: an RCT would like to explore as uniformly as possible, whereas each patient prefers ``exploitation'', \ie treatments that seem best. We incentivize participation by leveraging information asymmetry between the trial and the patients. We measure statistical performance via worst-case estimation error under adversarially generated outcomes, a standard objective for RCTs. We obtain a near-optimal solution in terms of this objective: an incentive-compatible mechanism with a particular guarantee, and a nearly matching impossibility result for any incentive-compatible mechanism. We consider three model variants: homogeneous patients (of the same ``type" comprising preferences and medical histories), heterogeneous agents, and an extension with estimated type frequencies.

% \begin{keywords}%
\xhdr{Keywords:}
exploration-exploitation tradeoff, incentive-compatibility, economic design, randomized controlled trials, multi-armed bandits
% \end{keywords}
\end{abstract}

% hides the sections in body from from ToC
\addtocontents{toc}{\protect\setcounter{tocdepth}{0}}

\section{Introduction}
\label{sec:intro}
Randomized Controlled Trials (RCTs) are gold-standard ways to evaluate the safety and efficacy of medical treatments.
%\footnote{We use ``RCT" acronym in this sense throughout. Another common meaning, ``randomized \emph{controlled} trial'', is also consistent with our model, since one of the arms there is typically a ``control".}
A paradigmatic design considers $n=2$ alternatives, a.k.a. \emph{arms}.%
\footnote{While one arm is typically a ``control group'' (such as an old treatment or a placebo), this distinction is not essential throughout this paper.}
Each patient~$t$ is assigned an arm $a_t$ independently and uniformly at random. An outcome $\outcome_t$ is observed and evaluated according to some predefined metric $f(\cdot)$. The goal is to use the observations to estimate the counterfactual averages
    $\tfrac{1}{T}\,\sum_{t\in [T]} f(\outcome_{a,t})$, $a\in[n]$,
where $T$ is the number of patients. Independent per-patient randomization over arms is essential to remove selection bias and ensure statistical validity. Uniform randomization is usually preferred, as it minimizes the estimation error. In particular, one can obtain strong provable guarantees even if each arm's outcomes are generated by an adversary, rather than sampled from some fixed (arm-dependent) distribution.%
\footnote{In the same sense as ``oblivious adversary" in adversarial bandits \citep{bandits-exp3}.}

A well-known issue that inhibits RCTs is the difficulty of recruiting patients \citep[\eg see surveys][]{ross1999barriers,mills2006barriers,rodriguez2021barriers}. This issue may be especially damaging for large-scale RCTs for widespread medical conditions with relatively inexpensive treatments. Then finding suitable patients and providing them with appropriate treatments would be fairly realistic, but incentivizing patients to participate in sufficient numbers may be challenging.

Similar issues arise in user-facing online platforms,
where experimentation is pervasive and RCTs, a.k.a.~A/B testing, are a standard approach.%
\footnote{For background on A/B testing, see \eg \citet{KohaviLSH09,KohaviAB-2015}.} While platform users are typically not informed of the particular experiments they are assigned to, the loss in performance (whether actual or perceived) may lead to user dissatisfaction. Having robust participation incentives would therefore be reassuring to the users, and may allow for more ambitious and risky experiments.

The standard approach of subsidizing the agents may be too expensive at the scale required. Indeed, clinical trials may require large subsidies when health outcomes are at stake, and online platforms may need a very large number of participants. The subsidies may also lead to selection bias towards poorer or more adventurous participants.

\xhdr{Our scope.}
We initiate the study of incentivized participation in RCTs via information asymmetry rather than subsidies. Our objective is to incentivize participation while optimizing statistical performance. We frame it as a non-standard variant of \emph{exploration-exploitation tradeoff}:
the RCT always prefers to \emph{explore} as uniformly as possible, whereas each patient/user prefers to \emph{exploit} (choose an arm that seems best), and might not participate if the RCT is inconsistent with this preference.%
\footnote{The standard tradeoff has an algorithm balance exploration vs exploitation to converge on the best arm.}
We henceforth refer to the potential RCT participants as \emph{agents}.

We leverage \emph{information asymmetry} between the RCT and the agents. We use preliminary data from the ongoing trial to skew the random choice of arms in favor of a better treatment (\ie towards exploitation), while retaining a sufficient amount of near-uniform exploration. This causes the agents to prefer participation in the RCT to the alternatives available externally.
%alternatives treatment options they might have (which also don't have  access to the trial's data).

We make two standard assumptions from economic theory: the RCT has the power to (i) prevent undesirable information disclosure to the agents, and (ii) commit to following a particular policy. These assumptions are realistic in clinical RCTs due to regulation and standard practice. Indeed, the history of an ongoing RCT and the assigned treatment should not be revealed to the patients (or their doctors) to the extent possible, to avoid biasing them in favor of  better-performing treatments \citep[pp. 26-30]{PCORI-adaptive-2012}.%
\footnote{In particular, different treatments are typically presented in the same way, \eg as externally-similar pills.}
 Moreover, the design of an RCT must be fixed in advance, spelling out all applicable decision points, and revealed to the patients \citep{CITI-GCP-2012}.
For online platforms, the data from an ongoing RCT is typically not revealed to the users, and the commitment power is a standard assumption in the literature (see Related Work).

%Indeed, it is a standard practice to not reveal the history of an ongoing RCT and the assigned treatment, to avoid biasing patients and doctors in favor of  better-performing treatments (\eg see \citet[pp. 26-30]{PCORI-adaptive-2012}). Moreover, the design of an RCT is typically required by regulation to be fixed in advance, spelling out all applicable ``decision points", and (essentially) it needs to be revealed to the patients to meet the standards of ``informed consent" \citep{CITI-GCP-2012}.

While RCTs with $n=2$ arms are the paradigmatic case, our results seamlessly extend to $n>2$. Note that clinical RCTs with $\geq 3$ arms are studied in biostatistics, \citep[\eg][]{Hellmich-multiarm-2001,Freidlin-multiarm-2008}, and are important in practice,   \citep[\eg][]{Lancet-multiarm-2014,Redig-BasketTrials-2015}.

\xhdr{Adversarial outcomes and time-invariant beliefs.}
% \asmargincomment{distilled \& bubbled up this discussion}
We aim for statistical guarantees that are as strong as those achievable by RCTs with per-agent randomization. Therefore, we allow the outcomes to be generated by an adversary (which motivates the need for per-agent randomization). In particular, our statistical guarantees are model-free once the data are collected.

For participation incentives, we aim for guarantees that are standard in the relevant literature on economic design. While we allow for adversarial outcomes,
we posit that the agents believe the outcome-generating process is much more well-behaved. Specifically, the agents believe that one's outcome distribution is determined by the assigned arm (and one's public type, if applicable), independent of the round. We call such beliefs \emph{time-invariant}. We posit that the beliefs are shared by all agents and are known to the mechanism, as an appropriately defined Bayesian prior. Time-invariant beliefs are standard in modeling sequential decision-makers (\eg in incentivized exploration, see Related Work).

This dichotomy between adversarial outcomes and simpler, time-invariant beliefs of the agents is essential to our model. While the time-invariant beliefs need not be correct, they are reasonable as a model of reality and quite realistic as beliefs that patients or platform users might have, particularly so for clinical RCTs where they represent pre-existing medical knowledge. Beliefs whose support does not include the actual problem instance (a.k.a. misspecified beliefs) are not uncommon in economic theory, see Related Work.

\xhdr{Our model.}
We model ``incentivized participation" as a mechanism design problem under a statistical objective and incentive constraints. An RCT (\emph{mechanism}) sequentially interacts with self-interested agents like a multi-armed bandit algorithm: in each round, an agent arrives, the mechanism chooses an arm for this agent, and observes an outcome. The mechanism's statistical objective is standard for RCTs: counterfactually estimate each arm's performance on the whole population, even if the outcomes are generated by an adversary.%
\footnote{In contrast, the standard objective for multi-armed bandits is to optimize cumulative payoffs.}
While there are several reasonable ways to formalize this objective, we choose one that is most convenient for our purposes: the worst-case mean-squared estimation error, denoted $\err$ (and formally defined in \Cref{sec:model-stats}).
%Formally, the mechanism optimizes the worst-case estimation error for these counterfactual estimates, denoted $\err$.
If not for the incentive constraints for participation, the mechanism would simply choose arms independently and uniformly at random for each agent.

The role of $f$, the RCT's metric, is as follows. Typically, $f(\outcome)$ corresponds to the ``societal value" of the outcome $\outcome$. The RCT may also be interested in the empirical frequency of $\omega$, which is captured by setting $f = \ind{\outcome}$, the indicator function that chooses this outcome. Our approach does not hinge upon a particular choice of $f$. In fact, our definition of $\err$ handles all metrics $f$ at once, and therefore applies to estimating societal value as well as outcome frequencies. In contrast, agents have subjective utilities over the outcomes that they wish to optimize, possibly different from the ``societal value" or any other metric used by the RCT.

We address participation incentives by imposing Bayesian Individual Rationality (BIR) constraints, a standard notion in economic theory.
Specifically, the mechanism commits to some randomized policy that assigns an arm to each agent, depending on the history of interactions with the previous agents. The agent does not observe the history or the chosen arm (as discussed above). The BIR constraint states that each agent must prefer participation to her ``outside option,'' and it (only) needs to hold with respect to agents' time-invariant beliefs. In particular, it suffices to provide participation incentives for the hypothetical world in which the agents' beliefs are correct. We note that a stronger version of participation incentives --- robust against an adversarially chosen outcome sequence --- is impossible to achieve (cf. \Cref{fn:beliefs}).

%While our statistical guarantees must be robust against adversarial outcomes, the outcome-generating process could be much more well-behaved. Crucially, the \emph{agents} believe that it is. Specifically, the agents believe that one's outcome distribution is determined by the assigned arm and one's observable attributes (resp., medical history or user's profile). While such beliefs may be incorrect, they appear necessary for tractability, and are particularly appropriate in clinical RCTs as a way to represent pre-existing medical knowledge.

We consider three model variants. To isolate the core issues of incentivized participation, we first consider the fundamental case of \emph{homogeneous agents} with the same beliefs and preferences.
Then, the general case of \emph{heterogeneous  agents} endows each agent with a \emph{type}. Then observable attributes such as medical history or user profile comprise one's \emph{public type} and determine one's beliefs over outcomes.
Moreover, one has unobservable preferences over the outcomes (\emph{private type}) which need to be elicited in an incentive-compatible way.%
\footnote{In practice, preference elicitation would happen during registration/opt-in process. For online platforms, private preferences and their elicitation make sense only if the RCT has explicit opt-in. The latter is uncommon, but may be needed if the RCT involves high-stakes interactions (such as buying a house) or is potentially controversial.}
The third model variant provides the mechanism with estimated type frequencies over the agent population.
%We leverage them to mitigate the harmful influence of types that are rare but inherently more ``difficult" to handle.
We address situations in which some types are inherently (much) more ``difficult" to handle than others. When such ``difficult" types are rare, we leverage the type frequencies to mitigate their harmful influence.

\xhdr{Our results.}
We focus on mechanisms with simple two-stage designs: a relatively short \emph{warm-up stage} of predetermined duration, followed by the \emph{main stage} that needs to be counterfactually estimated. The data from the warm-up stage (\emph{warm-up data}) creates information asymmetry in the main stage  which creates participation incentives, whereas the main stage determines the counterfactual estimates.
\footnote{The main stage may, in principle, also contribute to participation incentives. We found this capability unnecessary for our optimal mechanisms, but our lower bounds allow it.}
 RCTs with a warm-up stage
%that influences the main stage
are very common in practice: \eg phase-I clinical trials to calibrate the dosage and ``pilot experiments" in online platforms.

Our guarantees focus on the main stage, given enough warm-up data.
%The key point is that
A relatively small amount of warm-up data, determined by the prior and the utility structure, suffices to bootstrap the main stage for an arbitrarily large number of agents.
%assuming that enough data is collected in the warm-up stage. The aforementioned
The warm-up data can be collected exogenously \eg using paid volunteers, or endogenously by a mechanism in our model, \eg by leveraging
%Endogenous implementation of the warm-up stage can rely on
prior work on incentivized exploration (see ``Related work'' and \cref{lm:ec15}).

We obtain a near-optimal solution for each model variant described above, from homogeneous agents to heterogeneous agents to estimated type frequencies. First, we put forward an idealized benchmark which focuses on a single round of our problem. Then, we design a mechanism whose estimation error $\err$
(in our actual, multi-round model)
% \asmargincomment{we did not really define ``dynamic model" or "multi-round model" as terms.}
%\yledit{in the dynamic model}
is upper-bounded in terms of this single-round benchmark. Finally, we prove that no mechanism achieves better estimation error, up to constant factors.%
\footnote{The upper and lower bounds show that $\err$ is on the order of $B/T$, where $B$ is the benchmark. This is a different notion of optimality compared to bandits (where one usually optimizes regret, \ie the $o(T)$ additive term).}
 While the benchmark, the mechanism, and the lower bound are different for each model variant, many ideas carry over from one model variant to the next.

The single-round benchmark operates a hypothetical world in which the true ``state of the world'' is drawn according to the agents' beliefs. The benchmark is a hypothetical policy
that observes the true state (and the agent's type, if applicable) and chooses a distribution over actions so as to optimize a certain objective: essentially, an upper bound on $\err$. Thus, the benchmark encapsulates the optimal dependence on the prior in this hypothetical scenario.

Our mechanisms have a simple structure. First, their type-dependent ``sampling distributions" (from which the arms are sampled) do not adapt to the data collected in the main stage and do not change throughout the main stage. Second, counterfactual estimates are formed using the standard machinery of Inverse Propensity Scores (IPS). In contrast, our negative results apply to a wide class of mechanisms, allowing arbitrary dependence on the past data (in most rounds), arbitrary changes throughout the main stage, and arbitrary counterfactual estimates. Thus, our ``simple" mechanisms are proved near-optimal among arbitrary ``complex" mechanisms.

%The mechanism design (within the simple structure outlined above) is non-trivial for heterogeneous agents.

\xhdr{Significance.}
Our contribution spans the whole ``journey": from the new model to the benchmark to the mechanisms to the matching upper/lower bounds, and from homogeneous agents to heterogeneous agents to estimated type frequencies. In particular, the inherent complexity of each problem variant is expressed in terms of the respective single-round benchmark, a much simpler, ``static" object amenable to further analysis; modeling/definitions are absolutely crucial for that. The substance of our analysis is connecting the dynamic process and statistical estimators in the mechanism with a static information-design problem in the benchmark. We emphasize homogeneous agents to introduce the core model and ideas in a more accessible way. The general case requires a heavier setup, non-trivial mechanisms (within the simple structure outlined above), and a much more technical analysis.
% disentangling them from the considerable extra complexity  (ideas, technicalities, and heavier notation) needed for the general case.
% \asmargincomment{A version of this sentence used to be in the ``map".}

The non-adaptivity of our mechanisms is desirable in applications. Indeed, clinical RCTs in practice typically do not adapt to data, except possibly at a few pre-determined ``review points",%
\footnote{\label{fn:stages}This is because treating the patients and/or observing their outcomes may take a long time, and the doctors may find it difficult to coordinate the treatments across the patients. }
and neither do most RCTs in online platforms due to engineering considerations. In fact, many/most patients (resp., platform users) may need to be treated in parallel, and this is feasible for the main stage in our mechanisms. On the other hand, data-adaptive data collection allowed by our lower bounds is not out of scope, either, both for online platforms \citep[\eg][]{MWT-WhitePaper-2016,DS-arxiv} and for clinical RCTs (see Related Work). We conjecture it may help in addressing some of the open questions outlined in Conclusions.

We focus on the asymptotically-optimal analysis of the main stage, illuminating the inherent power and limitations of information-driven participation incentives. Our guarantees for the warm-up stage make a point (that a constant-sized warm-up suffices for any $T$) and emphasize interpretability, but we suspect they can be improved.

\xhdr{Additional discussion on clinical RCTs.}
%We emphasize that our contribution for a particular model variant is not the mechanism per se, but the whole ``ensemble", from the new model to the benchmark to the matching upper/lower bounds.
While our results are asymptotic, and hence most meaningful for RCTs with a large number of patients, participation incentives would make such trials more common and realistic. For concreteness, we provide
%also illustrate our results with simple
numerical examples (\cref{apx:examples}).

Theoretical participation incentives come at a price: the warm-up samples are discarded for counterfactual estimates, and the main-stage sampling distributions are skewed. This ``price of incentives" should be included in the overall  ``price of experimentation" incurred to obtain reliable medical knowledge. Whether it is acceptable in practice depends on how many patients these incentives actually help to recruit.
%how much these incentives actually help in recruiting patients.

%that the society pays for the sake of obtaining

We focus on patients' reluctance to participate that stems from their preferences among the available treatments (as opposed to other factors, \eg logistical difficulties). Further, this paper is only concerned with patients' initial decision to participate in an RCT. Their incentives to \emph{remain} in the trial, \eg if the treatment appears ineffective, are outside of our scope.%
\footnote{This is indeed an issue in practice, as per \citep[Section 2.1.7.2]{FDA-guidance-2001}. However, it is unclear if any approach based on information asymmetry could possibly help here: once the patient learns about the actual outcome, the “principal” does not have much/any informational advantage left.}

\xhdr{Map of the paper.}
We first present and analyze our model for the case of homogeneous agents (\cref{sec:model,sec:single}).
%This is to introduce the core model and ideas in a more accessible way, disentangling them from the considerable extra complexity  (ideas, technicalities, and heavier notation) needed for the general case.
%\cref{sec:single} works out the case of homogeneous agents, introducing key ideas with a more accessible exposition.
The main case of heterogeneous agents is defined and treated in \cref{app:worst}. The extension with estimated type frequencies is in \Cref{app:multi}. The technical machinery in \cref{sec:single,app:worst,app:multi} is introduced gradually and heavily reused from one section to the next.

\section{Related work}
\label{sec:related}

\noindent \textbf{Incentivized exploration (IE)} is a closely related model concerning experimentation on online platforms. Here the ``platform" is a recommendation system which
% Closest to our model is work on \emph{incentivized exploration} (\emph{IE}).
%motivated by recommendation systems. The platform
strives to incentivize users to \emph{explore} when they would prefer to \emph{exploit}. The platform issues recommendations (which the users can override), and makes them incentive-compatible by leveraging information asymmetry and the power to commit to a particular policy. Introduced in \citep{Kremer-JPE14,Che-13}, IE is well-studied, see \cite[Chapter 11]{slivkins-MABbook} for a survey. Most technically related are \citet{ICexploration-ec15,Jieming-multitypes18}.

The main difference concerns performance objectives. IE strives to maximize the total reward under a time-invariant outcome distribution, or just explore each arm once. Neither serves our statistical objective, which requires independent near-uniform randomization in each round: reward-maximization pulls a mechanism in a very different direction, and merely exploring each arm $N$ times does not ensure statistical validity under adversarially chosen outcomes.
%Typical objectives in IE are to maximize the cumulative reward (under time-invariant outcome distributions) and/or explore each arm once in a minimal number of rounds, both subject to incentive-compatibility. In contrast, our statistical  objective requires independent, near-uniform randomization in each round.  The desiderata of reward-maximization and uniform randomization pull a mechanism in two very different directions. Likewise, exploring each arm a given number of times does not suffice to ensure a statistically valid RCT unless one has independent per-round randomization or assumes time-invariant outcome/reward distribution.
The economic constraints are different, too, because agents in our model can only choose whether to participate.%
%patients cannot observe or alter their treatments (but can choose whether to participate).%
% whereas users in a recommendation system can.%
\footnote{This distinction is consequential only for $n>2$ arms. Then the economic constraint in IE is stronger than ours.}
%\asdelete{Mechanism's power to control information and commit to a particular policy are substantial assumptions in IE (\eg \citet{Bahar-ec16,Jieming-unbiased18} are written to mitigate them), but they are very benign for RCTs, as discussed above.}
%
Hence, our mechanisms, analyses and technical guarantees are very different from those in IE. However, our mechanisms reuse IE for the ``warm-up stage'' (see \cref{lm:ec15}).

A variant of IE creates incentives via payments
\citep[\eg][]{Frazier-ec14,Kempe-wine15}. However, this work also assumes time-invariant reward distributions and optimizes total reward. Moreover, it lets the agents observe full history, which is not the case for RCTs.

\xhdr{Economic design.}
\citet{narita2021incorporating} incentivizes participation in clinical RCTs when the patients are certain about the arms' payoffs and have a fixed outside option. Accordingly, information asymmetry (and learning over time to create it) plays no role in participation incentives. Certain beliefs are a degenerate case from our perspective (which allows and leverages distributional beliefs).

Using information asymmetry and commitment power to incentivize agents to take a particular action is a common goal in \emph{information design} \citep{BergemannMorris-survey19,Kamenica-survey19}. However, this literature tends to have a very different technical flavor.
The key distinction is that the designer's utility is directly derived from each agent's actions. In contrast, our statistical objective depends on the sequence of sampling distributions and does not add up across agents.

Misspecified beliefs (\ie beliefs whose support does not include the actual problem instance) are studied in a considerable recent literature in economic theory \citep[starting from][]{berk1966limiting,esponda2016berk}. However, this literature tends to focus on convergence and equilibria
% in repeated games, in terms of agents' belief or behavior
\citep[e.g.,][]{fudenberg2021limit,esponda2021equilibrium,esponda2021asymptotic},
%\citet{li2020misspecified} focus on the principal-agent problem where the principal can influence the agent's belief updating process by administrating the default action in every period.
whereas our paper takes the mechanism design perspective and is concerned with statistical estimation.

\citet{chassang2012selective} apply economic design to RCTs for a different purpose, focusing on agents' ability to strategically choose their effort level. They use principal-agent framework to disentangle the effects of treatment and effort without compromising external validity. Within this framework, incentives are created via monetary contracts.

%\yledit{Our idea of skewing the sampling probabilities of different arms for the design of trials also has a long history in economics \citep[e.g.,][]{zelen1969play,chassang2012selective,narita2021incorporating}. However, the models and design objectives of those papers are substantially different from ours. The main rationale for our unbalanced probabilities is to address the participation incentives, where those considerations are absent in their models, with the exception of \citet{narita2021incorporating}. Our distinction from \citet{narita2021incorporating} is that they assume that the principal is certain of the payoffs of the agents for various treatments, while in our model, that has to be learned over time by carefully designing the trials. }

%\xhdr{Statistics and machine learning.}
%Absent economic incentives, our problem reduces to the design of RCTs in bioinformatics and conterfactual estimation in machine learning. A detailed review of either literature is beyond our scope.  Data-adaptive designs for RCTs are not uncommon \citep[\eg see survey][]{Chow-adaptive-2008}. However, while prior work used data-adaptivity to improve the trial's statistical properties, we use it for participation incentives. Our mechanisms use a standard counterfactual estimator called \emph{inverse propensity score} (IPS). While more advanced estimators are known \citep[\eg][]{DR-StatScience14,Adith-nips15,AlekhMiro-icml17}, IPS is optimal in some worst-case sense (see \cref{sec:prelims}), which suffices for our purposes.

\xhdr{Clinical RCTs.}
We only point out the key connections to the (huge) medical literature on RCTs, as a more detailed review is beyond our scope.

Difficulties in patient recruitment are widely acknowledged as a significant barrier
\citep[\eg][]{adams1997recruiting,grill2010addressing}. The reasons are well-studied, \eg see surveys \citep{ross1999barriers,mills2006barriers,rodriguez2021barriers}. Patients' preferences among treatments are listed among the major factors, along with other factors such as insufficient trust, logistical difficulties, aversion to randomization, and loss of control. Prior work suggested increasing the probability of the ``new" treatment as a possible remedy
\citep{vozdolska2009net,jenkins2000reasons,pocock1979allocation,karlawish2008redesigning},
albeit without a formal model of participation incentives and without theoretical guarantees. These papers focus on the scenario when the new treatment appears ex-ante preferable compared to non-participation, and do not revise the treatment probabilities as new data arrive. Other suggested remedies focusing on eligibility criteria, patient outreach, trial logistics, etc., are not as relevant to this paper.

Data-adaptive designs for RCTs are not uncommon and have a long history in biostatistics, \eg see \citet{Chow-adaptive-2008} for a research survey, and \citet{PCORI-adaptive-2012,pallmann2018adaptive} for practitioners' guidance. While prior work used data-adaptivity to improve efficiency (the tradeoff between the trial's size and statistical properties), we use it to create participation incentives. Accordingly, the novelty in our RCT design is in connecting it to the analysis of incentives (and the corresponding lower bounds on the statistical objective), not in the RCT design per se. That said, our RCT for the homogeneous-agent case is based on a standard ``epsilon-greedy" approach from multi-armed bandits, whereas our design for heterogeneous agents appears new.

Our mechanisms are superficially similar to \emph{patient preference trials} (PPTs) \citep{Torgerson1998-PPT}, in that both allow patient preferences to impact treatment allocation. The key difference is that PPTs allow the patients to directly choose the treatment (potentially at the expense of internal validity), whereas in our setting a patient’s report only affects the sampling distribution (hence internal validity is not compromised). Another distinction is that we elicit patient preferences over the treatment outcomes, not over the treatments themselves.

\xhdr{Statistics and machine learning.}
Our mechanisms use a standard counterfactual estimator called \emph{inverse propensity score} (IPS). While more advanced estimators are known \citep[\eg][]{DR-StatScience14,Adith-nips15,AlekhMiro-icml17}, IPS is optimal in some particular worst-case sense (see \cref{sec:prelims} and \cref{app:stats-LB}), which suffices for our purposes.

While our mechanism follows the protocol of multi-armed bandits \citep{slivkins-MABbook,LS19bandit-book},
it strives to choose arms uniformly at random, whereas a bandit algorithm strives to converge on the best arm.%
\footnote{Under the standard objective of cumulative reward/loss optimization, as expressed by regret.}
More related is ``pure exploration"
\citep[starting from ][]{Tsitsiklis-bandits-04,EvenDar-icml06},
where a bandit algorithm predicts the best arm under a stationary outcome distribution, after exploring ``for free" (\eg uniformly) for $T$ rounds. In contrast, our model allows adversarially chosen outcomes and requires estimates for all arms and all outcomes. %(which is important in scenarios without a clear numerical objective).
Note that bandit algorithms, whether regret-minimizing or ``pure exploration",
are generally not geared to incentivize participation.
%\footnote{One exception is the \emph{greedy algorithm} (which always \emph{exploits}). However, it exhibits stark learning failures \citep{BSL-myopic23}, and does not randomize among arms (except for tie-breaking).}
%
Interestingly, both ``stationary" and adversarially chosen outcome distributions -- like, resp., agent beliefs and statistical objective in our model -- are well-studied in bandits, as, resp., ``stochastic" and ``adversarial" bandits.

%The \emph{greedy algorithm} (which always \emph{exploits}) is incentive-compatible with participation (if it uses the current agent's beliefs in each round). However, it exhibits stark learning failures \citep{BSL-myopic23}, and does not randomize among arms (except for tie-breaking).

%A notable exception is the \emph{greedy algorithm} which always \emph{exploits}, \ie chooses the best arm given the current observations and the initial beliefs. This is compatible with participation, provided that the current agent's beliefs are used in each round. The greedy algorithm exhibits stark learning failures \citep[][Ch.11]{BSL-myopic23,slivkins-MABbook}, but performs  well (as a bandit algorithm) under strong assumptions on context diversity and reward structure \citep{bastani2017exploiting,kannan2018smoothed,externalities-colt18}. These positive guarantees still fall short of our statistical objective, as they (i) assume stochastic bandits and (ii) imply convergence on the best arm, rather than near-uniform exploration. More conceptually, the greedy algorithm does not randomize among arms, which makes it inherently not suitable for our purposes.

%The ``main stage" of our mechanisms use a version of \emph{epsilon-greedy}, a standard (albeit suboptimal) algorithm for stochastic bandits which explores uniformly with some chosen probability and exploits otherwise. The exploration probability is maximized in a non-trivial way subject to agents' incentives, possibly depending on the agent's  type.

%\section{Our model: incentivized clinical trial}
\section{Our model: incentivized RCT (homogeneous agents)}
\label{sec:model}
%We first present our model for homogeneous agents, \ie agents with the same ``type", then extend it to heterogeneous agents in \cref{app:worst}. This is to disentangle the core model, which is quite complex even for a single agent type, from some subtlety and heavy notation needed for multiple types.
%

We focus on homogeneous agents, \ie agents with the same ``type". We consider a randomized controlled trial (RCT) with some patients/users, respectively  interpreted as a \emph{mechanism} and \emph{agents}. The mechanism interacts with the agents as a multi-armed bandit algorithm, but optimizes a very different statistical objective, and operates under economic constraints to incentivize participation. Accordingly, we split the model description into three subsections:
the \emph{bandit model} for the interaction protocol
% (\cref{sec:model-bandits}),
the \emph{statistical model} for its performance objective,  %(\cref{sec:model-stats}),
and the \emph{economic model} for the agents' incentives. % (\cref{sec:model-econ}).

%A reader can also skip \cref{sec:model-hetero} at the first reading and proceed to our solution for the homogeneous case (\cref{sec:prelims,sec:single}).

\subsection{Bandit model: the interaction protocol}
\label{sec:model-bandits}

We have $T$ agents and $n$ treatments, also called \emph{arms}. The set of agents is denoted by $[T] := \cbr{1 \LDOTS T}$,
and the set of arms is denoted by $[n]$.
The paradigmatic case is $n=2$ arms, but our results meaningfully extend to an arbitrary $n$. Choosing a given arm for a given agent yields an observable, objective \emph{outcome} $\outcome\in\outcomes$.
\footnote{\Eg whether and how fast the patient was cured, and with which side effects (if any). Or, how a platform user has interacted with the platform, in terms of clicks, views, sales, etc.}
 The set $\outcomes$ of possible outcomes is finite and known.

The mechanism interacts with agents as follows. In each round $t\in[T]$,
\begin{OneLiners}
\item[1.] a new agent arrives
%, with outcomes
%    $\rbr{\outcome_{a,t}\in\outcomes:\,a\in[n]}$
%that are unknown to the mechanism.

\item[2.] The mechanism chooses the \emph{sampling distribution} over arms, denoted by $\sampleD_t$,  samples an arm $a_t\in[n]$ independently at random from this distribution, and assigns this arm to the agent.

\item[3.] The outcome $\outcome_t\in\outcomes$ is realized and observed by the mechanism and the agent.
\end{OneLiners}

Each outcome $\outcome$ is assigned a \emph{score} $f(\outcome)\in [0,1]$, which may \eg represent this outcome's value. The \emph{scoring function}
    $f:\outcomes\to[0,1]$
is exogenously fixed and known (its role is explained in \Cref{rem:f-meaning}).
Note that the mechanism follows a standard protocol of \emph{multi-armed bandits}: in each round it chooses an arm $a_t$ and observes a numerical score $f(\outcome_t)$ for playing this arm.

The \emph{outcome table}
    $\rbr{\outcome_{a,t}\in\outcomes:\,a\in[n],\; t\in[T]}$
is fixed in advance (where round $t$ indexes rows), so that
    $\outcome_t = \outcome_{a_t,\;t}\in\outcomes$ for each round $t$.
We posit that this entire table is drawn from some distribution, which we henceforth call an \emph{adversary}.%
\footnote{This corresponds to ``oblivious adversary" in multi-armed bandits, as if there is an adversary who chooses the outcome table upfront. An analogous game-theoretic terminology is a non-strategic player called ``Nature".}
An important special case is \emph{stochastic adversary}, whereby each row $t$ of the table
%        $\rbr{\outcome_{a,t}\in\outcomes:\,a\in[n]}$, $t\in[T]$
is drawn independently from the same fixed distribution.

%\begin{remark}
%The agents cannot change the assigned arms. Indeed, an RCT participant cannot change her assigned treatment and cannot even observe what this treatment is.
%\end{remark}

\subsection{Statistical model: the performance objective}
\label{sec:model-stats}

% \ascomment{Subtle point: the adversary chooses reported types, rather than private types. Private types are irrelevant for statistical modeling, I think.}

The mechanism should estimate the average score of each arm \textbf{without any assumptions} on the adversary, \ie on how the outcomes are generated. Not relying on modeling assumptions is a standard desideratum in randomized controlled trials.

More precisely, the mechanism designates the first $T_0$ rounds as a \emph{warm-up stage}, a small-scale pilot experiment to create incentives, for some $T_0\leq T/2$ chosen in advance. We are interested in the average score over the remaining rounds, which constitute the \emph{main stage}.

Fix an adversary \adv. The average score of a given arm $a$ over the main stage is denoted by
\begin{align}\label{eq:adv-averages-new}
\textstyle
\advF(a)
    := \frac{1}{|S|}\,\sum_{t\in S} f(\outcome_{a,t}),
    \quad \text{where } S=\cbr{T_0+1 \LDOTS T}.
\end{align}
After round $T$, the mechanism outputs an estimate $\estF(a) \in [0,1]$ for $\advF(a)$, for each arm $a$.

\begin{remark}\label[remark]{rem:f-meaning}
%The meaning of the scoring function $f:\outcomes\to[0,1]$ is as follows.
We focus on the following two scenarios. First,
$\advF(a)$ is the frequency of a given outcome $\outcome^*\in\outcomes$ in the main stage; then the scoring function is
    $f(\outcome) = \ind{\outcome = \outcome^*}$.
Second, if $f(\outcome)$ is the outcome's value (for the agent and/or the society),
then $\advF(a)$ is the average value of arm~$a$. 
In particular, the \emph{average treatment effect} can be expressed as
$\advF(\text{treatment})-\advF(\text{control})$. 
We keep $f$ abstract throughout to handle both scenarios in a uniform way.
\end{remark}

The estimation error is defined as the largest mean squared error (MSE). Formally,
for mechanism $\mech$ and adversary $\adv$ we define
\begin{align}\label{eq:err-defn}
\err\rbr{\mech\mid\adv}
    = \max_{f:\outcomes\to[0,1],\; a\in[n]}
        \E\sbr{(\,\estF(a) - \advF(a)\,)^2},
%    = \E\sbr{\max_{\outcome\in\outcomes,\,a\in[n]}
%        \abs{\estF(a,\outcome) - \advF(a,\outcome) }},
\end{align}
where the $\max$ is over all scoring functions $f$ and all arms $a$, and  the expectation is over the randomness in the mechanism and the adversary.

%\begin{remark}
The full specification of our mechanism is that it chooses $T_0$, then follows the interaction protocol from Section~\ref{sec:model-bandits}, and finally computes the estimates $\estF(a)$, $a\in[n]$. Thus, the performance objective is to minimize \eqref{eq:err-defn} uniformly over all adversaries $\adv$. Note that \eqref{eq:err-defn} essentially requires the mechanism to handle all $f$ at once.
%\end{remark}

\begin{remark}
Without economic constraints, a trivial solution is to take $T_0=0$ and randomize uniformly over arms, independently in each round. Indeed, this is a solution commonly used in practice. A standard IPS estimator (see \Cref{sec:prelims}) achieves $\err(\mech\mid\adv)\leq n/T$.
\end{remark}

%Note that $\advF$, $\estF$ and $\err$ can be defined relative to an arbitrary subset $S\subset [T]$, and some of our arguments indeed hold for an arbitrary $S\subset [T]$. To emphasize such arguments, we put $S$ in the superscript, writing, resp., $\advF[S]$, $\estF[S]$ and $\err[S]$.

\subsection{Economic model: agents, beliefs and incentives}
\label{sec:model-econ}

The agents' information structure is as follows (as motivated in the Introduction).
\textbf{(i)}
Each agent $t$ does not observe anything about the previous rounds.
%Indeed, not revealing any information about an ongoing trial is strongly encouraged in practice \citep[\eg][]{CITI-GCP-2012}.
Consequently, agent $t$ does not observe the chosen distribution $\sampleD_t$ if it depends on the history.
\textbf{(ii)}
Each agent $t$ does not observe the assigned arm $a_t$.
%Indeed, RCTs typically present different treatments in an identical way, \eg as a pill that looks and feels exactly the same; this is strongly encouraged whenever possible \citep[\eg][]{CITI-GCP-2012}.
\textbf{(iii)}
Each agent knows the specification of the mechanism before (s)he arrives.
%In practical terms, the entity which runs the RCT (\eg a hospital) commits to a particular design of the trial, and this design is revealed to all patients. This property is typically enforced via regulation.

The mechanism must incentivize participation, \ie satisfy \emph{Individual Rationality (IR)}. We require Bayesian IR in expectation over the agents' beliefs over the adversaries.%
\footnote{\label{fn:beliefs}We invoke agents' beliefs because satisfying the IR property for every adversary (a.k.a. \emph{ex-post IR}) is clearly impossible, except for a trivial mechanism that always chooses the agent's outside option.}
Crucially, we posit that \textbf{agents' beliefs are over stochastic adversaries}, \ie the agents believe that a given treatment yields the same outcome distribution in all agents.
\footnote{\label{fn:misspecified}
We called such beliefs \emph{time-invariant} in the Introduction. These beliefs are \emph{misspecified}, meaning that their support might not include the actual adversary (which need not be stochastic). See Related Work for background.}
All agents share the same belief, a common simplification in economic theory.  Such beliefs represent the current medical knowledge relevant to the clinical trial, or the learnings from user interactions collected by the online platform (which the platform may choose to share with its users).

%\begin{remark}
%Agents have \emph{misspecified beliefs}, in the sense that their support does not necessarily include the actual adversary. In particular, the actual adversary need not be stochastic. Misspecified beliefs are studied in a considerable recent literature in economic theory, see \cref{sec:related}.
%\end{remark}

Let us specify the agents' beliefs formally. The outcome of each arm $a$ is an independent draw from some distribution $\truestate(a)$ over outcomes each time this arm is chosen. For better notation, we identify \emph{states} as mappings
    $\mystate:[n]\to \Delta_\Omega$
from arms to distributions over outcomes, and $\truestate$ as the true state.%
\footnote{We use a standard notation: $\Delta_\Omega$ is the set of all distributions over $\Omega$.}
Then, the agents believe that $\truestate$ is drawn from some Bayesian prior: distribution $\prior$ over states. We posit that $\prior$ is known to the mechanism, and assume that it has finite support.

An agent's utility for a given outcome $\outcome$ is determined by the outcome, denote it
    $\val(\outcome)\in [0,1]$.
Agent's expected utility for arm $a$ and state $\mystate$ is given by
    $ \Val_{\mystate}(a) := \E_{\outcome\sim\mystate(a)}\sbr{\val(\outcome)}.$
The \emph{utility structure} $u:\outcomes\to[0,1]$ is known to the mechanism and all agents.

We ensure that each agent weakly prefers participation to the ``outside option".
%To formalize the latter, we posit
One way to formalize the outside option is that an agent could refuse to participate in the mechanism and instead choose an arm (\eg a drug) that appears best based on the initial information. This achieves  Bayesian-expected utility
\begin{align}\label{eq:model-outside}
    \ValOut := \max_{a\in [n]}\; \E_{\mystate\sim\prior}\;\sbr{\Val_{\mystate}(a)}.
\end{align}
More generally, we accommodates an arbitrary, exogenously given $\ValOut\in [0,1]$. Such $\ValOut$ could represent the Bayesian-expected utility of using a competing online platform.

Finally, we can state the IR property that the mechanism is required to satisfy.
%\begin{definition}\label{def:BIR-homo}
The mechanism is \emph{Bayesian Individually Rational} (\emph{BIR}) if
%and any type
%    $\type \in \pubtypes\times\subtypes$ it holds that
%\begin{align}\label{eq:def:BIR-homo}
    $\E\sbr{\Val_{\truestate}(a_t)} \geq \ValOut$.
%\end{align}
for each round $t$. The mechanism is BIR on a given set $S$ of rounds if this holds for each round $t\in S$.
%\end{definition}

\begin{remark}\label[remark]{rem:model-outside}
%While \eqref{eq:model-outside} is a natural ``outside option" within our model, our analysis accommodates an arbitrary, exogenously given $\ValOut\in [0,1]$, and all results carry over with minimal modifications.
An arbitrary $\ValOut\in [0,1]$ allows us to model several effects relevant to clinical RCTs.
%\begin{itemize}
First, some medical treatments might not be available outside of a clinical trial. Then we could take the $\max$ in \eqref{eq:model-outside} only over the treatments that are.
Conversely, a standard treatment might not be available \emph{inside} the trial, \eg in a new treatment vs. placebo trials. Then the $\max$ in \eqref{eq:model-outside} should also include the standard treatment.
Finally, patients might assign positive utility to participation itself, \eg because they value contributing to science or they expect to receive higher-quality and/or less expensive medical care. This can be modeled (albeit in a rather idealized way) by decreasing $\ValOut$ in \eqref{eq:model-outside} by some known amount.

%\end{itemize}

\end{remark}

\subsection{Preliminaries}
\label{sec:prelims}

%This section is also accessible when specialized to homogeneous agents. For heterogeneous agents, we only need to recap that $\type_t,\obstype^{(t)}\in\types$ are, resp., the type and the observed type in round $t$, where $\types$ is the set of all agent types. We will use this notation, noting that for homogeneous agents, there is only one type ($|\types|=1$). No other machinery from \cref{sec:model-hetero} will be invoked.

The tuple
    $\history_t = \rbr{\sampleD_s,\, a_s,\,\outcome_s:\; \text{rounds } s<t}$
denotes the history collected by the mechanism before a given round $t$. Let $\probE[\cdot]$ denote the probability according to agents' Bayesian beliefs.

\xhdr{Per-round recommendation policy.}
A single round $t$ of our mechanism can be interpreted as a stand-alone ``recommendation policy" which inputs the $\history_t$, and outputs a recommended arm. Such policies are key objects in our analysis. We define them formally below.

We focus on a hypothetical world that operates according to agents' beliefs, whereby we have stochastic adversaries and the true state $\truestate$ is drawn from the prior $\prior$. The history $\history_t$ can be interpreted as a \emph{signal} that is jointly distributed with state $\truestate$. Indeed, the prior and the mechanism in the previous rounds determine a joint distribution $\jointP^{(t)}$ for the pair $(\truestate,\history_t)$.

We will need a more general definition in which we replace the history with an abstract \emph{signal} jointly distributed with $\truestate$. (The joint signal-state distribution is determined by the prior and the signal-generating process; we leave the latter outside of our definition.)

\begin{definition}\label[definition]{def:policy}
A \emph{signal} $\signal$ is a random variable which, under agents' beliefs, is jointly distributed with state $\truestate$:
    $(\truestate,\signal)\sim\jointP$,
where the joint prior $\jointP$ is determined by $\prior$. A recommendation policy $\policy$ with signal $\signal$ is a distribution over arms, $\policy(\signal)$, parameterized by the signal.
%We denote this distribution as $\policy(\signal)$.
%    where $\type = (\pubtype,\subtype)$.
\end{definition}

As a shorthand, the Bayesian-expected utility of policy $\policy$ is
\begin{align}\label{def:policy-shorthand}
\Val_{\prior}(\policy)
    := \E_{(\mystate,\signal)\sim\jointP}\;\;
    \E_{a\sim\policy(\signal)}\;
    \Val_{\mystate}(a).
\end{align}
%We denote \eqref{def:policy-shorthand} as
%    $\Val_{\prior}(\policy,\type)$
%when the joint distribution $\jointP$ is determined by the prior $\prior$.

%The notion of BIR carries over in a natural way:
A policy is called BIR if any agent weakly prefers following the policy compared to the outside option:
%\begin{align}\label{def:policy-BIR}
%\E_{(\mystate,\signal)\sim\jointP}\;\;
%    \E_{a\sim\policy(\pubtype,\,\subtype,\,\signal)}\;
%    \Val_{\mystate}\rbr{a,\type }
$\Val_{\prior}(\policy) \geq \ValOut$.
%\end{align}
The policy is called \emph{strictly BIR} if $\Val_{\prior}(\policy) > \ValOut$.

The mechanism can be represented as a collection of recommendation policies $\rbr{\policy_t:\,t\in[T]}$, where policy $\policy_t$ inputs $\history_t$ as a signal and outputs the sampling distribution as
%\begin{align}\label{eq:prelims-policy-t}
$\sampleD_t = \policy_t\rbr{\history_t}$
%\end{align}
for each round $t$.
The mechanism is BIR in a given round $t$ if and only if policy $\policy_t$ is BIR.

%\xhdr{Non-adaptivity.}
%Our results focus on mechanisms that do not adapt to new observations over some time period. Given a contiguous set $S$ of rounds, we say that a mechanism is \emph{non-data-adaptive} over $S$ if for each round $t\in S$ the sampling distribution $\sampleD_t$ does not depend on the data collected during $S$. If furthermore the mechanism uses the same recommendation policy $\policy_t$ throughout $S$, then we call it \emph{stationary} over $S$. In this terminology, all mechanisms that we design are stationary over the entire main stage. Our lower bounds merely posit that {the mechanism is not data-adaptive on some $S$} such that $|S|/T$ is an absolute constant.

\xhdr{Inverse-Propensity Scoring (IPS)}
is a standard estimator we use in all positive results.
%Fix a subset of rounds, $S\subset [T]$.
For each arm $a$, letting $S = \cbr{T_0+1 \LDOTS T}$ be the set of rounds in the main stage, the estimator is
%    $\ips:[n]\times\outcomes\to [0,1]$
\begin{align}\label{eq:IPS-defn}
\ips(a)
    := \textstyle  \frac{1}{|S|}\,\sum_{t\in S}
        \ind{a_t = a} \cdot f(\outcome_t)\,/\,\sampleD_t(a).
\end{align}

% \asdelete{IPS estimator is worst-case optimal.}
Consider a mechanism \mech which chooses a fixed sequence of sampling distributions
    $\rbr{\sampleD_t: t\in S}$.
If \mech uses the IPS estimator, $\estF = \ips$, then
\begin{align}\label{eq:IPS-UB}
\err\rbr{\mech\mid\adv}
    \leq \textstyle
        |S|^{-2}\;\max_{a\in[n]} \;\sum_{t\in S} 1/\sampleD_t(a)
    \qquad\text{for any adversary \adv}.
\end{align}
This guarantee is ``folklore",
%in fact, it is essentially optimal for any estimator $\estF$.
see \cref{app:IPS-UB} for a proof. This is the statistical tool we use to analyze our mechanisms. For the lower bounds, we use a somewhat non-standard tool which applies to arbitrary counterfactual estimators (\refeq{eq:IPS-LB-body} and \cref{app:stats-LB}).
\section{Homogeneous agents}
\label{sec:single}
This section treats the fundamental case of homogeneous agents. We obtain an optimal solution, with nearly matching upper/lower bounds, and introduce key ideas for the general case.

%For the ease of exposition, we state the main results first, in the cleanest possible way. A more detailed statement of the positive result is provided in \cref{sec:homo-details}, and the proofs are deferred to \cref{sec:homo-proofs-UB} and \cref{sec:homo-proofs-LB}.

\subsection{Benchmark}
We express our results in terms of a benchmark which focuses on a single round and encapsulates the optimal dependence on the prior. To this end, we consider recommendation policies that input the true state ($\signal=\truestate$); we call them \emph{state-aware}. We optimize among all BIR, state-aware policies $\policy$ so as to
%maximize the minimal sampling probability. More precisely, we
minimize the error upper bound \eqref{eq:IPS-UB}. More precisely, we'd like to minimize
    $\max_a 1/\policy_a(\truestate)$
in the worst case over adversaries \adv, \ie in the worst case over all possible realizations of $\truestate$.

Formally, the \emph{state-aware benchmark} is defined as follows:
\begin{align}\label{eq:bench-single-generic}
\bench(\prior)
     = \inf_{\text{BIR, state-aware policies $\policy$}} \quad
        \max_{\text{states }\mystate\in\support(\prior),\;\text{arms }a\in[n]} \quad
        \frac{1}{\policy_a(\mystate)}.
\end{align}

An alternative characterization of this benchmark is based on \emph{epsilon-greedy}, a well-known bandit algorithm captured (in our terms) by the next definition.

\begin{definition}\label[definition]{def:homo-greedy}
Fixing $\eps\geq 0$ and state $\mystate$, \greedy{\eps}{\mystate} is a recommendation policy that \emph{explores} with probability $\eps\geq 0$, choosing an arm independently and uniformly at random, and \emph{exploits} with the remaining probability,  choosing the best arm $\bestA$ for a given state $\mystate$.
\end{definition}

\noindent We use this policy with $\mystate=\truestate$ and the largest $\eps$ such that this policy is BIR.

\begin{lemma}\label[lemma]{lm:greedy-single}
The supremum in \refeq{eq:bench-single-generic} is attained by \greedy{\eps}{\truestate}, for some $\eps\geq 0$. Thus:
\begin{align}\label{eq:lm:bench-single}
\bench(\prior)
     = \inf_{\eps\geq 0:\;\greedy{\eps}{\truestate} \text{ is BIR}} \,n/\eps.
\end{align}
\end{lemma}

\noindent Note that the right-hand side of \eqref{eq:lm:bench-single} also depends on the prior $\prior$ through the definition of BIR.

\begin{remark}
Both formulations, \eqref{eq:bench-single-generic} and \eqref{eq:lm:bench-single}, are used in the analysis: the former for the lower bound and the latter for the upper bound. $\greedy{\eps}{\mystate}$ policy is also used in our mechanism. While epsilon-greedy is suboptimal as a regret-minimizing bandit algorithm, it suffices to achieve optimal performance for our purposes.
\end{remark}

\begin{proofof}[\cref{lm:greedy-single}]
The ``$\leq$" direction in \eqref{eq:lm:bench-single} holds because whenever policy
    $\greedy{\eps}{\truestate}$
is BIR, we have $\policy_a(\psi) \geq \eps/n$ for all arms $a$ and all states $\mystate$.
Focus on the ``$\geq$" direction from here on.
Fix $\delta > 0$. Let $\policy$ be a BIR, state-aware policy which gets within $\delta$ of the supremum in \eqref{eq:bench-single-generic}, so that
    $\policy_a(\mystate)\geq \frac{1}{\bench(\prior)+\delta}$
for all arms $a$ and all states $\mystate$.
Let $\policy'$ be the $(\eps,\truestate)$-greedy policy with
$\eps = \frac{n}{\bench(\prior)+\delta}$. Its expected utility is
\begin{align}
\Val_{\prior}(\policy')
&= \E_{\mystate\sim\prior}
    \sbr{(1-\eps +\eps/n)\cdot \Val_{\mystate}(\bestA)
    + \frac{\eps}{n}\cdot {\sum_{a\neq \bestA}} \Val_{\mystate}(a)} \nonumber \\
&\geq \E_{\mystate\sim\prior}
\sbr{\policy_{\bestA}(\mystate)\cdot \Val_{\mystate}(\bestA)
+ \policy_a(\mystate)\cdot {\sum_{a\neq \bestA}} \Val_{\mystate}(a)}
  \label{eq:lm:pf:bench-single}\\
&= \Val_{\prior}(\policy) \geq \ValOut. \nonumber
\end{align}
The inequality in \eqref{eq:lm:pf:bench-single} holds since $\policy_a(\mystate)\geq \eps/n$, so for each state the right-hand side shifts probabilities towards arms with weakly smaller expected utility. It follows that policy $\policy'$ is BIR. So, this policy gets within $\delta$ from the supremum in \eqref{eq:lm:bench-single}. This holds for any $\delta>0$.
\end{proofof}

The benchmark is finite as long as the best arm strictly improves over the outside option.%
\footnote{Failing that, we have a degenerate case when the agents believe that the trial is completely useless.}
To express this point, let $\bestA \in \argmax_{a\in [n]} \Val_{\mystate}(a)$ be \emph{the best arm}: a utility-maximizing action for a given state $\mystate$, ties broken arbitrarily. Its Bayesian-expected reward is
% \asmargincomment{new notation!}
%\begin{align}\label{eq:homo-bestV}
    $\bestVal := \E_{\mystate\sim\prior}\sbr{\Val_{\mystate}\rbr{\bestA}}$.
%\end{align}

\begin{claim}\label[claim]{cl:single-finite}
If
   $\bestVal>\ValOut$
then $\bench(\prior)<\infty$.
\end{claim}

\begin{proof}
If $\E_{\mystate\sim\prior}\sbr{\Val_{\mystate}\rbr{\bestA}}>\ValOut$,
there exists $\epsilon>0$ such that policy $\greedy{\eps}{\truestate}$
is BIR.
\cref{lm:greedy-single} then implies that the benchmark is finite.
\end{proof}

%We prove this claim in the next section, using the machinery developed therein. Note that the condition $\bestVal>\ValOut$ states that $\bestA[\truestate]$ is strictly BIR as a state-aware recommendation policy.

\subsection{Our mechanism: Main stage}

We focus on the main stage, coming back to the warm-up stage in \Cref{sec:warm-up}. We use the same recommendation policy in all rounds of the main stage.
%This recommendation policy is similar to \emph{epsilon-greedy}, a standard (but suboptimal) bandit algorithm.
Namely, we use policy \greedy{\eps}{\mystate},
%throughout the main stage,
for some fixed $\eps,\mystate$. Since the true state $\truestate$ is not known, we instead use state $\mystate = \approxstate$ which summarizes the data from the warm-up stage.
Specifically, let $\approxstate(a)$ be the empirical distribution over outcomes observed in the rounds of the warm-up stage when a given arm $a$ is chosen. Note that the average state $\approxstate(a)$ is not necessarily in $\support(\prior)$, but the \greedy{\eps}{\approxstate} policy is well-defined.
% This fully specifies the main stage, up to parameter $\eps$.

To recap, our mechanism uses \greedy{\eps}{\approxstate} recommendation policy in all rounds, for some $\eps>0$. (Unlike the standard epsilon-greedy algorithm, the estimated best arm is fixed throughout, by the choice of $\approxstate$, rather than adjusted as more data comes in.) Our guarantee is conditional on collecting enough samples in the warm-up stage.

\begin{theorem}\label{thm:mech-single}
Assume
   $\bestVal>\ValOut$.
Suppose mechanism \mech uses
    \greedy{\eps}{\approxstate}
policy in each round of the main stage, where
    $\eps = \tfrac{n}{2}/\bench(\prior)$,
and estimator $\estFreq = \ips$ as per \refeq{eq:IPS-defn}. Then:
\begin{flalign*}%\label{eq:thm:mech-single-main}
\;\;\;\text{(a)} \qquad & \err\rbr{\mech\mid\adv} \leq \frac{2\cdot \bench(\prior)}{T-T_0}
\qquad\text{for any adversary \adv}. &&
\end{flalign*}
\begin{itemize}
\item[(b)] Suppose  each arm
%(arm, public type) pair
appears in at least $N_{\prior}$ rounds of the warm-up stage, with probability $1$ over the agents' beliefs. Here $N_{\prior}<\infty$ is a parameter determined by the prior $\prior$, as per \refeq{eq:homo-N}. Then the mechanism is BIR over the main stage.

%\item[(c)] The mechanism is non-data-adaptive and stationary over the main stage, and uses IPS estimator (as defined in \refeq{eq:IPS-defn}) for estimator $\estFreq$.
\end{itemize}
\end{theorem}

The requisite number of warm-up samples, $N_{\prior}$, is expressed in terms of the \emph{prior gap},
\begin{align}\label{eq:prior-gap}
\gap(\prior) := \textstyle
\E_{\mystate\sim\prior}
\sbr{\Val_{\mystate}(\bestA) - \tfrac{1}{n}\cdot \sum_{a\in[n]}\Val_{\mystate}(a)},
\end{align}
the Bayesian-expected difference in utility between the best arm and the uniformly-average arm. Note that $\gap(\prior)>0$ provided that
$\bestVal>\ValOut$. Then
\begin{align}\label{eq:homo-N}
N_{\prior} = 32 \alpha^{-2} \log(8n/\alpha),
\quad \text{where }
    \alpha = n\cdot\gap(\prior)/\bench(\prior)>0.
\end{align}

We illustrate $\bench(\prior)$ and $N_{\prior}$ with a simple numerical example in \cref{apx:examples}.

%\begin{remark}
%One can improve the factor of 2 in \cref{thm:mech-single}(a) to $1+\delta$ for any fixed $\delta>0$. This requires setting
%    $\eps = \tfrac{n}{1+\delta}/\bench(\prior)$
%in \cref{thm:mech-single} and increasing  $N_{\prior}$ depending on $\delta$. However, $\bench(\prior)/(T-T_0)$ is merely an \emph{upper bound} on the best possible MSE, which can often be surpassed by replacing IPS with more advanced statistical estimators. So, meeting this upper bound is not as important, and we go with the factor of ``2" for ease of presentation.
%\end{remark}

%\subsection{Proofs: positive results}
%\label{sec:homo-proofs-UB}

%\xhdr{Analysis.}

\begin{proofof}[\cref{thm:mech-single}]
Let's apply concentration to the approximate state $\approxstate$. By Azuma-Hoeffding inequality, for any state $\mystate$ and any $\delta> 0$ we have
\begin{align*}%\label{eq:single-concentration}
\probE\sbr{|\,\Val_{\approxstate}(a) - \Val_{\mystate}(a)\,|\geq \delta \mid \truestate=\mystate}
\leq \proberr := 2n\cdot\exp\rbr{-2\delta^2 N_{\prior}}
\quad\forall a\in[n].
\end{align*}
So, one can compare the best arm $\bestA$ for the true state $\truestate=\mystate$ and the best arm $\bestA[\approxstate]$ for $\approxstate$:
%be the optimal action for state $\mystate$, it follows that with probability at least $1-\proberr$,
\begin{align}\label{eq:approx opt a}
\probE\sbr{ \Val_{\mystate}(\,\bestA[\approxstate]\,)
    \geq \Val_{\mystate}(\,\bestA\,) - 2\delta \mid \truestate=\mystate} \geq 1-\proberr.
\end{align}

Next, we show that \mech is BIR in the main stage. Henceforth, let
    $\delta = \tfrac{n}{8}\cdot\gap(\prior)/\bench(\prior)$.
%\asdelete{and $N_{\prior} = \tfrac{32}{(\bench(\prior)\cdot \gap(\prior))^2} \log\tfrac{8n}{\bench(\prior)\cdot \gap(\prior)}$}
The agent's expected utility in a single round of the main stage is
\begin{align*}
\Val_{\prior}\rbr{\greedy{\eps}{\approxstate}}
&=\E_{\mystate\sim\prior}
    \sbr{(1-\eps)\cdot \Val_{\mystate}(a_{\approxstate})
    + \frac{\eps}{n}\cdot {\sum_{a\in[n]}} \Val_{\mystate}(a)}\\
&\geq  \E_{\mystate\sim\prior}
    \sbr{(1-\eps)\cdot \Val_{\mystate}(\bestA)
    + \frac{\eps}{n}\cdot {\sum_{a\in[n]}} \Val_{\mystate}(a)}
    - \proberr - 2\delta\\
&=  \E_{\mystate\sim\prior}
    \sbr{(1-2\eps)\cdot \Val_{\mystate}(\bestA)
    + \frac{2\eps}{n}\cdot {\sum_{a\in[n]}} \Val_{\mystate}(a)}
    + \eps\cdot\gap(\prior) - \proberr - 2\delta\\
&\geq \Val_{\prior}(\policy^*),
\end{align*}
where we denote policy \greedy{2\eps}{\truestate} with $\policy^*$.
The first inequality holds by \eqref{eq:approx opt a} (and the fact that utilities are at most $1$).
The last inequality holds because
    $\eps\cdot\gap(\prior) - \proberr - 2\delta\geq 0$
by our choice of $\eps,\delta$.

To complete the BIR proof, since $2\eps = n/\bench(\prior)$, policy $\sigma^*$ is precisely a $(n/\bench(\prior), \truestate)$-greedy policy. According to Lemma~\ref{lm:greedy-single}, this policy is BIR and $\Val_{\prior}(\policy^*)\geq \ValOut$.

Finally, recall that in each round of the main stage, each arm is sampled with probability at least $\eps/n = 1/(2\cdot \bench(\prior))$. Thus, \cref{thm:mech-single}(a) follows from \refeq{eq:IPS-UB}.
% and the fact that $T_0\leq T/2$,
\end{proofof}

\subsection{Our mechanism: Warm-up stage}
\label{sec:warm-up}

We leverage prior work on incentivized exploration to guarantee warm-up samples. Restated in our notation, the relevant guarantee is as follows:
% \citep[e.g.,][]{ICexploration-ec15}. %(see Appendix~\ref{apx:IE}).

\begin{lemma}[\citealp{ICexploration-ec15}]\label[lemma]{lm:ec15}
There exists a number $M_{\prior}\leq \infty$ which depends only on the prior $\prior$ (but not on the time horizon $T$) such that:
\begin{OneLiners}
\item[(a)] if $M_{\prior}<\infty$ then there exists a BIR mechanism which samples all arms in $M_{\prior}$ rounds.
\item[(b)] if $M_{\prior}=\infty$ and $n=2$ arms then no strictly BIR mechanism can sample both arms.
\end{OneLiners}
\end{lemma}

%\begin{corollary}\label{thm:mech-single-warm-up}

\noindent While \Cref{lm:ec15}(a) is trivial as stated above, \citet{ICexploration-ec15} provide an explicit mechanism for part (a) and an explicit expression for $M_{\prior}$. The latter is somewhat tedious, and not essential for this paper.%
\footnote{However, let us present it for $n=2$ arms, for the sake of completeness. Consider $\Delta_{\mystate} := \Val_{\mystate}(1)-\Val_{\mystate}(2)$, the ``gap" between the arms' expected utilities. W.l.o.g., assume that
    $\E_{\mystate}\sbr{\Delta_{\mystate}}\geq 0$,
\ie arm $1$ is (weakly) preferred initially. Let
    $X_{(k)} := \E_{\mystate}\sbr{-\Delta_{\mystate} \mid S_{(k)}}$,
where $S_{(k)}$ is a tuple of $k$ independent reward samples from arm $1$. Then, suppose there exists some finite $k = k_{\prior}$ such that
    $\Pr\sbr{ X_{(k)}>0 }>0$
(if not, set $M_{\prior} = \infty$).
Let $Y = X_{(k_{\prior})}$. Then
\begin{align*}
M_{\prior} = (n+1) \max\rbr{k_{\prior},\; \frac{\E_{\mystate}\sbr{\Delta_{\mystate}}}
    {\E\sbr{Y\mid Y>0}\; \Pr[Y>0] }}.
\end{align*}
}
Thus, a sufficient warm-up stage for \cref{thm:mech-single} lasts
    $T_0 = M_{\prior}\cdot N_{\prior}$
rounds, and consists of $N_{\prior}$ consecutive runs of the mechanism from \Cref{lm:ec15}(a).

\begin{remark}
While \citet{ICexploration-ec15} assumes stochastic adversaries, this suffices because the BIR property is required only for  time-invariant beliefs. So, the mechanism from \Cref{lm:ec15}(a) is BIR regardless of whether the actual adversary is stochastic. The mechanism samples all arms in $M_{\prior}$ rounds deterministically, regardless of the observed outcomes.
\end{remark}

The warm-up stage can be replaced with a non-BIR mechanism that incentivizes agent participation by other means, \eg by paying the volunteers, as commonly done in medical trials. Crucially, $n\cdot N_{\prior}$ volunteers suffice to bootstrap our mechanism to run for an arbitrarily large $T$.

\subsection{Lower bound}

We prove that the guarantee in \cref{thm:mech-single}(a) is essentially optimal in the worst case.

Our result applies to all mechanisms that do not adapt to new observations over some constant (but possibly small) fraction of rounds, and are unrestricted otherwise. Formally, given a contiguous set $S$ of rounds, we say that a mechanism is \emph{non-data-adaptive} over $S$ if for each round $t\in S$ the sampling distribution $\sampleD_t$ does not depend on the data collected during $S$.

%We prove that the guarantee in \cref{thm:mech-single}(a) is essentially optimal, in the worst case over the adversaries. Our result applies to all mechanisms that are non-data-adaptive over some constant (but possibly small) fraction of rounds, but are unrestricted otherwise.

\begin{theorem}\label{thm:LB-single}
Fix prior $\prior$ and time horizon
    $T\geq \Omega\rbr{\bench(\prior)^{-2}}$
and $T_0\leq T/2$. Let $S$ be a contiguous subset of rounds in the main stage, with $|S|\geq c\cdot T$ for some absolute constant $c>0$. Fix mechanism \mech that is non-data-adaptive over $S$. Then some adversary $\adv$ has
\begin{align}\label{eq:thm:LB-single}
 \err\rbr{\mech\mid\adv}
    = \Omega\rbr{\frac{\bench(\prior)}{T}}.
\end{align}
\end{theorem}

Thus, \Cref{thm:LB-single} allows per-round policies that change over time and can themselves be jointly chosen at random. In rounds $t\not\in S$, the policies can depend on all past observations. Importantly, this includes ``multi-stage" mechanisms that partition the time into a constant number of contiguous ``stages" and are non-data-adaptive within each (possibly with a data-adaptive warm-up stage of duration, say, $<T/2$). Most medical trials in practice fit this model, see \cref{fn:stages}.

The proof plan is as follows. First, we derive a statistical lower bound for an arbitrary estimator (\refeq{eq:IPS-LB-body}). Second, we focus on a particular ``bad state" and ``bad arm" for the benchmark. Third, we construct a ``bad" adversary, using this bad state/arm pair. The final computation demonstrates that this adversary gives a large enough error. We flesh out this plan below. 

\xhdr{Proof of \cref{thm:LB-single}.}
%For ease of presentation, posit that $S$ is the entire main stage.
Fixing the history of the first $t_0 = \min(S)-1$ rounds, the mechanism draws the tuple of sampling distributions
    $(\sampleD_t: t\in S)$,
perhaps jointly at random. We prove that
% the following lower bound on the estimation error (see \cref{sec:homo-proofs-LB})
\begin{align}\label{eq:IPS-LB-body}
\err\rbr{\mech\mid\adv^\dag}
    \geq
    \Omega\rbr{\frac{1}{|S|^2}\;\max_{a\in[n]}\;\sum_{t\in S} \min\cbr{1/\E\sbr{\sampleD_t(a)}, \sqrt{|S|}}},
\end{align}
for some adversary $\adv^\dag$ determined by the expected sampling probabilities $\rbr{\E[\sampleD_t]:\,t\in S}$. This is a general statistical tool, see \cref{app:stats-LB} for a standalone formulation, discussion, and proof.

Given state $\mystate$, let $\adv_{\mystate}$ be the adversary restricted to the first $t_0$ rounds that samples outcome $\outcome_{a,t}$ from distribution $\mystate(a)$, independently for each arm $a$ and each round $t\leq t_0$. Recall that $\adv_{\mystate}$ and \mech jointly determine the sampling distributions $\sampleD_t$ for all rounds $t>t_0$. Let $N_\mystate(a)$ be the number of times arm $a$ is selected during the rounds in $S$ under adversary $\adv_{\mystate}$.

We claim there exists a state $\mystate_0\in\supp(\prior)$ and arm $a_0$ such that
\begin{align}\label{eq:thm:LB-single:pf-state-body}
    \E\sbr{N_{\mystate_0}(a_0)} \leq |S|/\bench(\prior).
\end{align}
% at any time $t$,
% there exists an action $a$ such that  the expected sampling probability for each action $a\in[n]$
Suppose not.
%    $\E\sbr{N_{\mystate}(a)} > \bench(\prior) \cdot (T-T_0)$
%for all arms $a$ and states $\mystate\in\supp(\prior)$.
% then there exists a time $t$ in the main stage such that the recommendation policy $\policy$ recommends all actions with probability exceeding $\bench(\prior)$.
% In this case,
Then we can construct a state-aware policy $\policy$ which  contradicts \eqref{eq:bench-single-generic}. Specifically, policy $\policy$ simulates a run of \mech under adversary $\adv_{\truestate}$, chooses a round uniformly at random from the main stage, and recommends the same arm as \mech in this round. Such policy is BIR, and therefore would make the right-hand side of \eqref{eq:bench-single-generic} strictly larger than $\bench(\prior)$, contradiction. Claim proved. From here on, we fix state $\mystate_0$ and arm $a_0$ which satisfy \eqref{eq:thm:LB-single:pf-state-body}.

% by first sampling $T_0$ outcomes according to state $\mystate$
% and then follow the recommending of policy $\policy$
% based on the simulated samples.
% By construction, policy $\policy'$ is BIR
% and the min sampling probability of $\policy'$ for the benchmark problem exceeds $\bench(\prior)$, which leads to a contradiction.

The adversary $\adv$ is constructed as follows. We use the adversary $\adv_{\mystate_0}$ for the first $t_0$ rounds, and adversary $\adv^\dag$ afterwards. It immediately follows from \eqref{eq:IPS-LB-body} that
\begin{align}
\err\rbr{\mech\mid\adv}
&\geq
    \Omega\rbr{\frac{1}{|S|^2}
    \sum_{t\in S}\;\max\cbr{\tfrac{1}{q_t},\; \sqrt{|S|}}},
    \qquad \text{where } q_t := \E\sbr{\sampleD_t(a_0)}
    \nonumber \\
&\geq
    \Omega\rbr{\frac{1}{|S|^2}\;
    \rbr{\sqrt{|S|}\cdot |\loS| + \sum_{t\in \hiS}\;\tfrac{1}{q_t}}},
    \label{eq:thm:LB-single:pf-immediate-body}
\end{align}
where
    $\loS = \cbr{t\in S:\; q_t\leq 1/\sqrt{|S|}}$
and
    $\hiS = \cbr{t\in S:\; q_t> 1/\sqrt{|S|}}$.
In words, $\loS$ (resp., $\hiS$) comprises the rounds in the main stage with low (resp., high) expected sampling probability $q_t$.

In the remainder of the proof, we analyze the right-hand side of \eqref{eq:thm:LB-single:pf-immediate-body}. This is an expression in terms of deterministic scalars $q_t\in [0,1]$,  $t\in S$. We leverage the fact that \eqref{eq:thm:LB-single:pf-state-body} holds, so that
\[  \hiN := \sum_{t\in \hiS} q_t
        \leq \sum_{t\in S} q_t
        = \E\sbr{N_{\mystate_0}(a_0)} \leq |S|/\bench(\prior).\]

By harmonic-arithmetic mean inequality,
%Since a harmonic mean lower-bounds the corresponding arithmetic mean,
\[
\sum_{t\in\hiS} \frac{1}{q_t}
    \geq |\hiS|^2/\hiN
    \geq |\hiS|^2\cdot\bench(\prior)/|S|.
\]
Plugging this back into \eqref{eq:thm:LB-single:pf-immediate-body} and noting that  $|\loS|+|\hiS| = |S|\geq T/2$,
we have
\begin{align*}
\err\rbr{\mech\mid\adv}
&\geq
    \Omega\rbr{\frac{1}{|S|^2}\;
        \rbr{\sqrt{|S|}\cdot |\loS| + |\hiS|^2\cdot\bench(\prior)/|S|}}
        \geq \Omega\rbr{\bench(\prior)/T}.
\end{align*}
%The last inequality holds because
 %   $|\loS|+|\hiS| = |S|\geq T/2$.

%\section{Heterogeneous agents: an overview}
%\label{sec:hetero-overview}
%\input{sec-overview}

\section{Heterogeneous agents}
\label{app:worst}
We now consider the general case of heterogeneous agents.  Agent heterogeneity is two-fold: different agent types may have different outcome distributions for the same arm (according to the agents' beliefs), and different subjective utilities for the same outcome (which are not directly observable). We model this as follows: \textbf{public types determine outcome distributions} and the \textbf{private types determine subjective utilities}.
%\footnote{So, each agent's type consists of the public type and the private type; only the former is known to the mechanism.}
The intuition is that beliefs about objective outcomes are driven by the patient's objective medical history, which is typically well-documented and available to the clinical trial, whereas patients' (dis)utilities over medical outcomes can be highly subjective. For online platforms, we posit that users' beliefs about their own objective outcomes are driven by their attributes that are visible to the platform (via user profiles, cookies, etc.), whereas the respective utilities (such as value-per-click) are subjective and unobservable.
%, \eg in evaluating the disutility of particular medical complications.
%\asdelete{These subjective preferences are not known to the mechanism. In our model,}
Each agent reports its private type upon arrival, and the mechanism must incentivize truthful reporting.
%\asdelete{(in addition to the BIR property).}%
%\footnote{We express this new requirement in a Bayesian framework, and term it \emph{Bayesian Incentive-Compatibility}, or \emph{BIC}.}

\subsection{Overview}

We obtain matching upper and lower bounds like in \pref{sec:single}, but with a new benchmark, mechanism, and analysis. This novelty is necessitated by the Bayesian Incentive Compatibility (BIC) constraint which incentivizes truthful reporting of private types and binds the types together.%
\footnote{Indeed, without private types one can use mechanism from \cref{sec:single} tailored to each particular type.}
We define a new benchmark, denoted $\worstbench$, which generalizes \eqref{eq:bench-single-generic} to a given set $\types$ of agent types. Compared to the expression in \eqref{eq:bench-single-generic}, this benchmark requires all {BIC and BIR} per-round policies under the $\inf$,
% to additionally input a type and satisfy BIC,
and takes the $\max$ over $\types$ (as well as over all states). Our mechanism attains the same error bound as in \pref{thm:mech-single} relative to $\worstbench$, with a similar but much more technical sufficient condition for BIR. The lower bound carries over from \cref{thm:LB-single}, replacing $\bench(\prior)$ with $\worstbench$.

Our mechanism is more complex compared to \cref{sec:single}. Exploration probabilities now depend on the arm and the agent type, and are coordinated across types due to the BIC constraint. We specify them indirectly, via a state-aware policy that optimizes the benchmark. The true state $\truestate$ is estimated via the maximum likelihood estimator (MLE) given the data from the warm-up stage.

In terms of proof techniques, our mechanism requires more subtlety compared to \cref{sec:single} so as to analyze the MLE estimator. The lower bound is proved using the same techniques; the novelty here lies in formulating the definitions so that the theorem and techniques indeed carry over.

\subsection{Model: heterogeneous agents}
\label{sec:model-hetero}

We denote agent's type as $\type = (\pubtype,\subtype)$, where $\pubtype\in\pubtypes$ is the public type, and $\subtype\in\subtypes$ is the private type. The set of all possible types,
    $\types = \pubtypes\times\subtypes$,
is known to the mechanism. We assume that it is finite. The case of homogeneous agents is when
    $|\types|=1$.

Below we spell out the necessary changes to the model, compared to \cref{sec:model}.

\xhdr{Bandit model.} In the interaction protocol in Section~\ref{sec:model-bandits}, step 1 is modified as follows:
\begin{OneLiners}
\item[1a.]
A new agent arrives, with type
    $\type_t = (\pubtype^{(t)},\subtype^{(t)})$
and outcomes
    $\rbr{\outcome_{a,t}\in\outcomes:\,a\in[n]}$.

\item[1b.] The mechanism observes the public type $\pubtype^{(t)}$, but not the private type $\subtype^{(t)}$ or the outcomes.

\item[1c.]
The agent reports its private type, not necessarily truthfully, as \emph{reported type} $\reptype^{(t)}\in\subtypes$.
\end{OneLiners}
The pair
    $\obstype^{(t)} = \rbr{\pubtype^{(t)},\, \reptype^{(t)}}$
observed by the mechanism will be called the \emph{observed type}.

The agents' information structure is the same as in the homogeneous case, with one addition: agents know their own types. An additional economic constraint defined below (\cref{def:BIC}) ensures \emph{truthful reporting} of the private types:
    $\reptype^{(t)} = \subtype^{(t)}$
for all rounds $t$.

The adversary now specifies the agents' types, in addition to the outcomes. Formally,
the table
    $\rbr{ \type_t, \outcome_{a,t}\in\outcomes:\,a\in[n]}_{t\in [T]}$
is fixed in advance (where $t$ indexes rows). The table drawn at random from some distribution called an \emph{adversary}. A \emph{stochastic adversary} is defined as before: it draws each row $t$ from some fixed distribution.

The history now includes types: the tuple
    $\history_t = \rbr{\type_t,\, \sampleD_s,\, a_s,\,\outcome_s:\; \text{rounds } s<t}$
denotes the history collected by the mechanism before a given round $t$, assuming truthful reporting of agents' types in the previous rounds.

\xhdr{Statistical model.}
Given an adversary \adv, the average score $\advF(a)$ and the estimation error $\err(\mech\mid\adv)$ are defined exactly the same as in the homogeneous case, as per \refeq{eq:adv-averages-new} and \refeq{eq:err-defn}. The estimation error is defined assuming truthful reporting.

\begin{remark}
Under truthful reporting, the adversary controls the
types % everything
observed by the mechanism (because reported types are just private types, which are selected by the adversary). An alternative interpretation for the statistical model, not contingent on truthful reporting, is that the adversary controls reported types directly (and then the private types are  irrelevant).
\end{remark}

%We obtain stronger results when the empirical type frequencies are approximately known in advance.
%Empirical type frequencies for a given adversary \adv are defined as
%\begin{align}\label{eq:type-rates}
%\textstyle \typeFreq(\theta \mid \adv)
%    = \tfrac{1}{|S|}\,\sum_{t\in S}\;
%    \ind{(\pubtype^{(t)},\, \reptype^{(t)}) = \type}
%        \qquad\forall\type \in \pubtypes\times\subtypes.
%\end{align}

\xhdr{Economic model.} The definitions are extended to include the agents' types, as per the semantics described above. In what follows, consider an agent with type
    $\type = (\pubtype,\subtype)\in \pubtypes\times\subtypes$.

Agents' beliefs are extended as follows. When some arm $a$ is chosen, the outcome is an independent draw from some distribution $\truestate(a,\pubtype)$ over outcomes; note that the outcome distribution is determined by both the arm and the public type. Mappings
    $\mystate:[n]\times\pubtypes\to \Delta_\Omega$
are identified as \emph{states}, and~$\truestate$ as the true state.
As before, the agents believe that $\truestate$ is drawn from some Bayesian prior $\prior$ over finitely many possible states.%
\footnote{Prior $\prior$ represents common pre-existing medical knowledge, so it is reasonable for it to be same for all agents.}
In each round $t$, type $\type_t$ is drawn independently from some fixed distribution $\typeDist$.
Both $\prior$ and $\typeDist$ are known to the mechanism and the agents.

To simplify presentation, we assume that all outcome distributions $\mystate(a,\pubtype)$ have  full support over the set $\outcomes$ of possible outcomes, \ie that this holds for all arms $a$, public types $\pubtype$, and states $\mystate\in\support(\prior)$; we term this the \emph{full support assumption}.%
\footnote{This assumption can be removed, at the cost of some technical complications, see \Cref{rem:full-support}.}

The agent's \emph{utility} for a given outcome $\outcome$ is $\val(\outcome,\subtype)\in [0,1]$; note that it is determined by both the outcome and the private type.
The realized utility is subjective and not directly observable. The \emph{utility structure}
    $\val:\outcomes\times\subtypes\to[0,1]$
is known to the mechanism and all agents.

Thus, the agent's expected utility and outside option are redefined as
\begin{align*}
\Val_{\mystate}(a,\type)
    &= \E_{\outcome\sim\mystate(a,\pubtype)} \val(\outcome,\subtype)
    \quad \text{for arm $a$ and state $\mystate$},\\
\ValOut(\type)
    &= \max_{a\in[n]}\; \E_{\mystate\sim\prior}\sbr{\Val_{\mystate}\rbr{ a,\type }}.
\end{align*}

\begin{remark}
Our analysis also accommodates an arbitrary known $\ValOut: \types\to[0,1]$, like in the homogeneous case, see \Cref{rem:model-outside}. All results carry over with minimal modifications.
\end{remark}

%\yldelete{, except  \Cref{cl:hetero-finite,cl:hetero-freq-finite}}

To extend the BIR property, we focus on Bayesian-expected utility from truthful reporting. Let $\truE{t}$ be the event that all agents $s\leq t-1$ have reported their private type truthfully:
    $\reptype^{(s)} = \subtype^{(s)}$.
Let
    $\agentE\sbr{\cdot \mid \mE}$
denote expectation with respect to agents' Bayesian beliefs given event $\mE$.%
\footnote{We use non-standard notation for the expectation to emphasize that it only needs to be well-defined under event~$\mE$, even though it may be ill-defined unconditionally.}

\begin{definition}\label[definition]{def:BIR}
The mechanism is BIR if for each round $t$,
\begin{align}\label{eq:def:BIR}
    \agentE\sbr{\Val_{\truestate}\rbr{a_t,\type}
        \mid \truE{t} \text{ and } \type_t = \type}
    \geq \ValOut(\type)
    \qquad\text{for any type $\type \in \types$}.
\end{align}
The mechanism is called BIR on a given set $S$ of rounds if \eqref{eq:def:BIR} holds for each round $t\in S$.
\end{definition}

\begin{remark}
We require the BIR property to hold for all agent types, even though it may be more difficult to satisfy for some types than for some others. Indeed, if some agent type chooses not to participate, the score estimates $\estF$ computed by the mechanism would likely suffer from selection bias. Technically, the BIR property can be ensured w.l.o.g. for every agent (at the expense of estimation quality) simply by choosing this agent's outside option deterministically.
\end{remark}

We also require the mechanism to (weakly) incentivize truthful reporting of private types, as captured by the following definition.

\begin{definition}\label[definition]{def:BIC}
The mechanism is called \emph{Bayesian Incentive Compatible}  (\emph{BIC})
if for any round~$t$, type
    $\type\in\types$,
and private type $\subtype'\in\subtypes$ it holds that
\begin{align}\label{eq:def:BIC}
\agentE\sbr{\Val_{\truestate}\rbr{a_t,\type}
    \mid \truE{t},\; \type_t = \type}
\geq \agentE\sbr{\Val_{\truestate}\rbr{a_t,\type}
    \mid \truE{t},\; \pubtype^{(t)} = \pubtype,\; \reptype^{(t)} = \subtype'}.
\end{align}
The mechanism is called BIC on a given set $S$ of rounds if \eqref{eq:def:BIC} holds for each round $t\in S$.
\end{definition}

Given a BIC mechanism, it is a (weakly) best response for each agent to report the private type truthfully, assuming the previous agents did so. We assume that all agents report truthfully. 
%Note that restricting to BIC mechanisms is without loss of generality according to the revelation principle \citep{Myerson}.

\begin{remark}
Private type elicitation does not have to be BIC. However, mechanisms in our model need to interpret (mis)reported private types, one way or another, and the  BIC constraint is one standard way to proceed.
Note that restricting to BIC mechanisms and truthful reporting is without loss of generality according to Myerson's revelation principle \citep{Myerson}.
\end{remark}

\xhdr{Recommendation policies}
are redefined to also input an observed type. Formally, a recommendation policy $\policy$ with signal $\signal$ is a distribution over arms, parameterized by the signal and the observed type $\type\in\types$. We denote this distribution as $\policy(\signal,\type)$. (Recall that signal is defined as per \Cref{def:policy}.) Accordingly, the Bayesian-expected utility of policy $\policy$ is redefined as
\begin{align}\label{def:policy-shorthand-heterogeneous}
\Val_{\prior}(\policy,\type)
    := \E_{(\mystate,\signal)\sim\jointP}\;\;
    \E_{a\sim\policy(\signal,\,\type)}\;
    \Val_{\mystate}\rbr{a,\type },
\quad\text{where }
    \type\in\types.
\end{align}
%We denote \eqref{def:policy-shorthand} as
%    $\Val_{\prior}(\policy,\type)$
%when the joint distribution $\jointP$ is determined by the prior $\prior$.

The notions of BIR and BIC carry over in a natural way. The policy is called \emph{BIR} if any agent weakly prefers reporting truthfully and following the policy compared to the outside option:
\begin{align}\label{def:policy-BIR}
%\E_{(\mystate,\signal)\sim\jointP}\;\;
%    \E_{a\sim\policy(\pubtype,\,\subtype,\,\signal)}\;
%    \Val_{\mystate}\rbr{a,\type }
\Val_{\prior}(\policy,\type)
\geq \ValOut(\type)
\qquad\forall
    \type %= (\pubtype,\subtype)
        \in \types.
\end{align}

\noindent The policy is called \emph{BIC} if for any type
    $\type = (\pubtype,\subtype) \in \types$
and any observed type
    $\obstype = (\pubtype,\reptype) \in \types$
it holds that
\begin{align}\label{def:policy-BIC}
%\E_{(\mystate,\signal)\sim\jointP}\;\;
%    \E_{a\sim\policy(\pubtype,\,\subtype,\,\signal)}\;
%    \Val_{\mystate}\rbr{a,\type }
\Val_{\prior}(\policy,\type)
\geq
\E_{(\mystate,\signal)\sim\jointP}\;\;
    \E_{a\sim\policy\rbr{\signal,\,\obstype}}\;
    \Val_{\mystate}\rbr{a,\type}.
\end{align}

The mechanism can be represented as a collection of recommendation policies $\rbr{\policy_t:\,t\in[T]}$, where policy $\policy_t$ inputs history $\history_t$ as a signal and determines the sampling distribution as
\begin{align}\label{eq:prelims-policy-t-heterogeneous}
\sampleD_t = \policy_t\rbr{\history_t,\,\obstype^{(t)}}
    \quad\text{for each round $t$}.
\end{align}
The mechanism is BIR  (resp., BIC) in a given round $t$ if and only if policy $\policy_t$ is BIR  (resp., BIC).
%for joint prior $\jointP^{(t)}$.
%(Recall that the joint prior $\jointP^{(t)}$ assumes truthful reporting in the previous rounds.)

\subsection{Our results}

%We now consider the general case of heterogeneous agents, with multiple public and private types.

We obtain matching upper and lower bounds like in \pref{sec:single}, but with more technicalities and new benchmark, mechanism, and analysis. This novelty is necessitated by the presence of multiple private types and the BIC condition which binds the types together. Indeed, without private types (formally, $|\subtypes|=1$), the BIC condition vanishes, and the public types are not bound to one another by the BIR condition alone. Then one can treat each public type separately, and use the mechanism from \cref{sec:single} tailored to this type.

\xhdr{Benchmark.} Our bounds are expressed in terms of the following benchmark:
\begin{align}\label{eq:worst-bench}
\worstbench
     = \inf_{\text{BIC, BIR, state-aware $\types$-policies $\policy$}} \;\;
        \max_{\mystate\in\support(\prior),\,a\in[n]}\;\;
        \max_{\text{types }\type\in\types} \;
        \frac{1}{\policy_a(\mystate,\type)}.
\end{align}
Here, $\types$-policies are recommendation policies with type set $\types$, as defined above.

This benchmark generalizes \eqref{eq:bench-single-generic}, the benchmark for the homogeneous case, taking the $\max$ over all types inside the expression. In particular, \eqref{eq:bench-single-generic} can be written as $\bench(\prior,\{\type\})$, where $\type$ is the unique type in the homogeneous case. In fact, the general benchmark admits a lucid alternative characterization via the worst type:
\begin{align}\label{eq:worst-bench-alt}
\worstbench
     = \max_{\type\in\types} \bench(\prior,\{\type\}).
\end{align}
However, our analysis entirely relies on the original formulation \eqref{eq:worst-bench}.

\begin{proofof}[\refeq{eq:worst-bench-alt}]
Denote the right-hand side of \refeq{eq:worst-bench-alt} with $\worstbenchhat$. First, it is easy to observe that $\worstbenchhat\leq\worstbench$, since $\worstbenchhat$ only requires BIR without incentive constraints. Next, we show that for any BIR policy $\policy$ with worst-case benchmark value $\worstbenchhat$, we can design a BIR and BIC policy that matches the same worst-case benchmark value.

In particular, consider another policy $\hat{\policy}$ which for any type $\type\in\types$, offers the utility-maximizing option among the ``per-type" policies $\policy(\,\cdot,\type')$, $\type'\in\types$. That is, for all states $\mystate$ and types $\type$,
\begin{align*}
\hat{\policy}(\mystate,\type)
:= \policy(\mystate,\type^*(\type)),
\text{ where }
\type^*(\type) \in
    \argmax_{\type'\in\types}
        \Val_{\prior}\rbr{\policy(\,\cdot\,,\type'),\type},
\end{align*}
with ties broken in some fixed way. Policy $\hat{\policy}$ is BIC by design, because
$\Val_{\prior}\rbr{\hat{\policy},\type}
       \geq \Val_{\prior}\rbr{\policy(\,\cdot\,,\obstype),\type}$
for any observed type $\obstype$.
%because all types pick the utility maximizing option from the same set.
Plugging in $\obstype = \type$, we see that
%Moreover, it is BIR because
\begin{align*}
\Val_{\prior}(\hat{\policy},\type)
%= \max_{\policy' \in \{\policy(\mystate,\type')\}_{\type'\in\types}}
%\Val_{\prior}(\policy',\type)
\geq \Val_{\prior}(\policy,\type)
\geq \ValOut(\type),
\end{align*}
where the last inequality holds by the BIR property of $\policy$. So, policy $\hat{\policy}$ is BIR.
Finally, the minimum sampling probability is the same for both policies:
$$\min_{\mystate\in\support(\prior),\,a\in[n],\, \text{types }\type\in\types}\;\;
\policy_a(\mystate,\type)
= \min_{\mystate\in\support(\prior),\,a\in[n],\, \text{types }\type\in\types}\;\;
\hat{\policy}_a(\mystate,\type).
 \qquad\qedhere $$

% \begin{align}
% \hat{\policy}(\mystate,\type)
%     = \policy(\mystate,\type^*),
% \text{ for some fixed type }
%     \type^* \in \argmax_{\type'\in\types}
%         \Val_{\prior}(\policy,\type').
% \end{align}
% Policy $\hat{\policy}$ is BIC simply because it does not depend on the input type $\theta$. It is BIR because
% \begin{align*}
% \Val_{\prior}(\hat{\policy},\type)
%     = \Val_{\prior}(\policy,\type^*)
%    \geq \Val_{\prior}(\policy,\type)
%    \geq \ValOut(\type)
%    \text{ for any type $\type\in\types$}.
% \end{align*}
% Here the first inequality holds by definition of $\hat{\policy}$, and the second inequality holds by the BIR property of policy $\policy$.
% %It is easy to verify that policy $\hat{\policy}$ is BIC and BIR.
% Finally, the minimum sampling probability is the same for both policies:
% $$\min_{\mystate\in\support(\prior),\,a\in[n],\, \text{types }\type\in\types}\;\;
% \policy_a(\mystate,\type)
% = \min_{\mystate\in\support(\prior),\,a\in[n],\, \text{types }\type\in\types}\;\;
% \hat{\policy}_a(\mystate,\type).
%  \qquad\qedhere $$
\end{proofof}

%Accordingly, we refer to \eqref{eq:worst-bench} as the \emph{worst-type benchmark}.

To state the suitable non-degeneracy conditions, let us extend the best-arm notation from \cref{sec:single}. Fix type $\type = (\pubtype,\reptype)\in\types$.
Let
    $\bestA[\mystate,\type] \in \argmax_{a\in [n]} \Val_{\mystate}(a,\type)$
be the best arm for a given state $\mystate$. The Bayesian-expected reward of the best arm is then
\begin{align}\label{eq:hetero-bestV}
\bestVal(\type)
    &:= \E_{\mystate\sim\prior}
            \sbr{\Val_{\mystate}\rbr{\bestA[\mystate,\type],\,\type }}
\end{align}
If the best arm is computed for the reported type $\reptype\in\subtypes$ (which is not necessarily reported truthfully), its Bayesian-expected reward becomes
\begin{align}\label{eq:worst-bestV-rep}
\bestVal(\type,\reptype)
    &:= \E_{\mystate\sim\prior}
            \sbr{\Val_{\mystate}\rbr{\bestA[\mystate,\type'],\,\type }},
            \quad\text{where }\type' = (\pubtype,\reptype).
\end{align}
Thus, the non-degeneracy condition is:
\begin{align}\label{eq:worst-nondegen}
\bestVal(\type) > \ValOut(\type)
    \text{ and }
\bestVal(\type) > \bestVal(\type,\reptype)
\qquad\forall\type\in\types,\,\reptype\in\subtype.
\end{align}

\begin{remark}\label{rem:worst-bestA}
We can interpret the best arm $\bestA[]=\bestA[\truestate,\type]$ as a state-aware recommendation policy, called the \emph{best-arm policy}. Then condition \eqref{eq:worst-nondegen} states that $\bestA[]$ is strictly BIR and strictly BIC.
\end{remark}

\begin{claim}\label[claim]{cl:hetero-finite}
Assume \eqref{eq:worst-nondegen}. Then $\worstbench<\infty$.
\end{claim}

\begin{proof}
By \eqref{eq:worst-nondegen},
the \greedy{\eps}{\truestate} policy is BIC and BIR for some $\eps>0$.%
\footnote{As in \cref{sec:single}, \greedy{\eps}{\truestate} chooses the best-arm policy in the exploitation branch.}
This policy satisfies the conditions under the $\inf$ in \eqref{eq:worst-bench}, and has a finite ``benchmark value", namely $\tfrac{1}{\eps}$.
\end{proof}

\xhdr{Lower bound.} The lower bound carries over from \cref{thm:LB-single}, and is proved using the same techniques. The novelty is in formulating the definitions so that the theorem and the techniques indeed carry over. We defer the proof to \cref{apx:LB}.

\begin{theorem}\label{thm:worst-LB}
The statement in \cref{thm:LB-single} holds for benchmark $\worstbench$.
\end{theorem}

\xhdr{Positive results.}
Our mechanism for the main stage is more complex compared to the one in  \cref{sec:single}. Exploration probabilities now depend on the arm and the agent type, and are coordinated across types due to the BIC constraint. We specify them indirectly, via a state-aware policy that optimizes the benchmark. {Specifically, let $\benchpolicy$ be a state-aware policy which optimizes the $\inf$ in \refeq{eq:worst-bench}; such policy exists as an optimizer of a continuous function on a compact space. The true state $\truestate$ is estimated via the maximum likelihood estimator (MLE) given the data observed in the warm-up stage,
denote it by $\estimatestate \in \supp(\prior)$.}
% \asmargincomment{moved the text around}

Our mechanism is defined as follows. We choose uniformly between $\benchpolicy$ and the best-arm policy $\bestA[]$. Formally: in each round $t$ of the main stage, the sampling distribution $\sampleD_t$ is the average between two distributions:
    $\benchpolicy(\estimatestate,\type)$
and
    $\bestA[\estimatestate,\type]$,
where
    $\type = \rbr{\pubtype^{(t)},\,\reptype^{(t)}}$
is the observed type. For estimator $\estFreq$, we use the IPS estimator \eqref{eq:IPS-defn}. This completes the description of the main stage. We denote our mechanism as $\mechname(T_0,\benchpolicy,\bestA[])$.

\begin{theorem}\label{thm:mech-hetero}
Assume \eqref{eq:worst-nondegen}. Let \mech be the mechanism $\mechname(T_0,\benchpolicy,\bestA[])$. Then:
\begin{flalign*}%\label{eq:thm:mech-single-main}
\;\;\;\text{(a)} \qquad & \err\rbr{\mech\mid\adv} \leq \frac{2\cdot \bench(\prior,\types)}{T-T_0}
\qquad\text{for any adversary \adv}. &&
\end{flalign*}
\begin{itemize}

\item[(b)] Suppose the following holds for the warm-up stage:
%collects at least $N_{\prior}$ samples of each type-arm pair, with probability at least $1-\initProb$ over the agents' beliefs,
\begin{align}\label{eq:thm:mech-single-main-prob}
\probE
%    \sbr{\text{at least $N_{\prior}$ samples of each type-arm pair are collected}}
    \sbr{\text{each (arm, public type) pair appears in $\geq N_{\prior}$ rounds}}
    \geq 1-\initProb,
\end{align}
here
    $N_{\prior}<\infty$ and $\initProb\geq 0$
are some parameters determined by the prior $\prior$ and the utility structure. Then the mechanism is BIR and BIC over the main stage.

\end{itemize}
\end{theorem}

\begin{remark}
The structure of $\benchpolicy(\truestate,\cdot)$, the optimal state-aware policy, is not essential for our analysis. Nonetheless, we state it below for the sake of clarity. Consider a ``menu" of per-type optimal policies, which are the $(\eps, \truestate)$-greedy policies with tight BIR constraints for the respective type. Then each type is mapped to the policy from this menu that maximizes its utility. This mapping ensures truthful reporting of private types.
\end{remark}

%The rest of this section focuses on the positive result, providing more details and a proof.

%Some mechanism \mech inputs $\delta>0$, satisfies the following:
%\begin{align}\label{eq:lm:mech-multi-worst}
%\err\rbr{\mech\mid\adv} \leq \frac{2\cdot\worstbench+2\delta}{T-T_0}
%\qquad\text{for any adversary \adv}.
%\end{align}

\subsection{More details for the positive result (\cref{thm:mech-hetero})}
\label{sec:hetero-details}

The parameters in part (b) of \cref{thm:mech-hetero} are driven by the \emph{degeneracy gap}: the gap in the non-degeneracy condition \eqref{eq:worst-nondegen}, defined as follows:
\begin{align}\label{eq:worst-nondegen-gap}
\eta^* = \min\rbr{
    \min_{\type\in\types} \bestVal(\type) - \ValOut(\type),\;
    \min_{\type\in\types,\, \reptype\in\pubtypes}
        \bestVal(\type) - \bestVal(\type,\reptype)}.
\end{align}
Another key parameter, $\lr_{\min}(\prior)>0$, measures the difficulty of learning the true state $\truestate$ via MLE. We defer its definition to \refeq{eq:worst-keyParam}. Then $N_{\prior}$ is defined as follows:
\begin{align}\label{eq:worst-N}
N_{\prior}(\eta^*)
    = 1+\frac{1}{2\lr_{\min}(\prior)}\cdot \log\frac{4\,|\supp(\prior)|}{\eta^* \cdot \lr_{\min}(\prior)}.
\end{align}

\begin{theorem}\label{thm:worst-UB-details}
\cref{thm:mech-hetero} holds with parameters $\initProb = \eta^*/8$, where $\eta^*$ is the non-degeneracy gap  from \eqref{eq:worst-nondegen-gap},
and $N_{\prior} = N_{\prior}(\eta^*)$, as per \refeq{eq:worst-N}.
\end{theorem}

The specification of $\lr_{\min}(\prior)$ is rather lengthy, as it requires some notation regarding  log-likelihood ratios (LLR). Fix two distinct states $\mystate,\mystate'\in\support(\prior)$. For any arm $a$, type $\type = (\pubtype,\subtype)\in\types$ and outcome $\outcome\in\outcomes$, let
\begin{align*}%\label{eq:hetero-LLR}
\llr_{a,\type,\outcome}\rbr{\mystate,\mystate'}
    := \log\frac{\mystate_{\outcome}(a,\pubtype)}{\mystate'_{\outcome}(a,\pubtype)}
\end{align*}
be the LLR between the two states. The corresponding KL-divergence is
\begin{align*}%\label{eq:hetero-KL}
\KL_{a,\type}\rbr{\mystate,\mystate'}
:=  \E_{\outcome\sim\mystate(a,\pubtype)}
            \llr_{a,\type,\outcome}\rbr{\mystate,\mystate'}
=   \KL\rbr{\mystate(a,\pubtype),\;\mystate'(a,\pubtype)}.
\end{align*}
Let is define the maximum change in LLR compared to its expectation,
\begin{equation}\label{eq:hetero-llmax}
\ll_{\max}\rbr{\mystate,\mystate'}
:=  \max_{a\in[n],\; \type\in\types,\;\outcome\in\outcomes}\;
    \abs{\llr_{a,\type,\outcome}\rbr{\mystate,\mystate'} - \KL_{a,\type}\rbr{\mystate,\mystate'}}
< \infty.
\end{equation}
The minimum positive KL-divergence is defined as
\begin{equation}\label{eq:hetero-klmin}
\KL_{\min}\rbr{\mystate,\mystate'}
:= \min_{a,\type} \KL_{a,\type}\rbr{\mystate,\mystate'}>0,
\end{equation}
where the $\min$ is over all pairs $(a,\type)\in [n]\times\types$ such that
    $\KL_{a,\type}\rbr{\mystate,\mystate'}>0$.
At least one such $(a,\type)$ pair must exist because the two states are distinct. Finally:
\begin{equation}\label{eq:worst-keyParam}
\lr_{\min}(\prior)
:= \min_{\text{states }\mystate,\mystate'\in\supp(\prior)}\;
        \rbr{\frac{\KL_{\min}(\mystate,\mystate')}{\ll_{\max}(\mystate,\mystate')}}^2
>0.
\end{equation}

\xhdr{Main stage: extensions.}
It may be desirable to use relaxed versions of the benchmark-optimizing policy $\benchpolicy$ and the best-arm policy $\bestA[]$. Indeed, the exact versions may be too difficult to compute or too complex to implement in practice; also, the prior $\prior$ or the utility structure may be not fully known to the mechanism.
We relax the benchmark-optimizing policy by allowing additive slack $\delta>0$, which shows up in the final regret bound. We relax the best-arm policy as an arbitrary state-aware policy $\hatpolicy$ which is BIR and BIC by some additive margin $\eta\leq \eta^*$; this margin only affects the number of warm-up samples.%

\begin{definition}
A recommendation policy is called \emph{$\eta$-BIR} (resp., \emph{$\eta$-BIC}) for some $\eta\geq 0$ if satisfies \refeq{def:policy-BIR} (resp., \refeq{def:policy-BIC}) with $+\eta$ added to the right-hand side.
\end{definition}

\begin{theorem}\label{thm:worst-UB-extension}
Assume \eqref{eq:worst-nondegen}. Let $\hatbenchpolicy$ be a state-aware policy which optimizes the $\inf$ in the benchmark \eqref{eq:worst-bench} up to an additive factor $\delta\geq 0$.
 Let $\hatpolicy$ be a state-aware policy which is $\eta$-BIR and $\eta$-BIC, for some $\eta\in (0,\eta^*]$. Then mechanism $\mechname(T_0,\hatbenchpolicy,\hatpolicy)$ satisfies part (b) in \cref{thm:mech-hetero}
with parameters
$\initProb = \eta/8$
and
    $N_{\prior} = N_{\prior}(\eta)$, and satisfies
\begin{align}\label{eq:thm:worst-UB}
\err\rbr{\mech\mid\adv} \leq \frac{2\,\worstbench + 2\delta}{T-T_0}
\qquad\text{for any adversary \adv}.
\end{align}
\end{theorem}

We recover \cref{thm:mech-hetero} since the best-arm policy is $\eta^*$-BIR and $\eta^*$-BIC by \eqref{eq:worst-nondegen-gap}.

%It may suffice to collect samples from only a \emph{subset} of type-arm pairs if there is a known correlation among them. (More specifically, if each feasible state $\mystate\in\supp(\prior)$ encodes such correlation.) The preliminary stage only needs to ensure that the estimation error bound (\cref{cl:hetero-UBpf-errorProb}) holds in the analysis. Such improvement is particularly valuable when some types occur very rarely and are difficult to incentivize.
%\ascomment{can you/we add a concrete simple example?}
%\yledit{
%Consider the case with binary actions $\{0,1\}$ and binary outcomes $\{0,1\}$.
%Suppose there are two public types $\pubtype, \pubtype'$
%such that probability of type $\pubtype$ is $1-\epsilon$ given the prior.
%Moreover, there are two possible state of the world $\mystate,\mystate'$ with uniformly random probability.
%Suppose with probability $\sfrac{1}{2}$, $\mystate(1, \pubtype)$ generates outcome $1$ with probability $1$
%while $\mystate'(1, \pubtype)$ generates outcome $1$ with probability $0$.
%With the rest probability, $\mystate(1, \pubtype)$ generates outcome $1$ with probability $0$
%while $\mystate'(1, \pubtype)$ generates outcome $1$ with probability $1$.
%In this example, it is sufficient to collect a single sample for public type $\pubtype$ given action $1$ to correctly estimate the state.
%}

\medskip

\begin{remark}\label{rem:full-support}
The ``full support" assumption stated in \cref{sec:model-hetero} can be removed, let us outline how. First, our analysis extends to a relaxed version of this assumption, when for any given pair $(a,\pubtype)\in [n]\times\pubtypes$, all outcome distributions $\mystate(a,\pubtype)$, $\mystate\in\support(\prior)$ have the same support, but this support may depend on the $(a,\pubtype)$ pair. Second, if the relaxed assumption is violated, then outcome distributions $\mystate(a,\pubtype)$ and $\mystate'(a,\pubtype)$ have different supports, for some states $\mystate,\mystate'\in\support(\prior)$, some arm $a$, and some and public type $\pubtype$. More specifically, there is an outcome $\outcome$ that lies in the support of $\mystate(a,\pubtype)$, but not in the support of $\mystate'(a,\pubtype)$. Then with sufficiently many samples from the
    $(a, \pubtype)$
pair, one would observe outcome $\outcome$ with high probability as long as it is in the support of
    $\truestate(a,\pubtype)$.
This would rule out either the case $\truestate=\mystate$ or the case $\truestate=\mystate'$. We take $N_\prior$ large enough for this to happen for every combination of $\outcome$, $a$, $\pubtype$, $\mystate$, $\mystate'$. Consider the set $\mystates_0$ of all states in $\support(\prior)$ that are not ``ruled out" after the warm-up stage; with high probability, the relaxed assumption holds with respect to this set.
\end{remark}

\xhdr{Warm-up stage}
is treated separately, as in \cref{sec:single}, either endogenously (via a BIR, BIC mechanism), or exogenously (\eg via paid volunteers). The sufficient amount of warm-up data depends only on the agent beliefs and the utility structure, whereas it bootstraps our mechanism to run for an arbitrarily large time horizon $T$. Endogenous data collection has been studied in prior work on incentivized exploration, for both public and private types: in some sense, one can explore all type-arm pairs that can possibly be explored \citep{ICexploration-ec15,Jieming-multitypes18}.

%\xhdr{Warm-up stage.}
%Warm-up data collection is treated separately, as in \cref{sec:single}. It can be done either endogenously (via a BIR, BIC mechanism), or exogenously (\eg by paying the volunteers). Either way, the sufficient amount of warm-up data depends only on the prior and the utility structure, whereas it bootstraps our mechanism to run for an arbitrarily large time horizon $T$.

%Endogenous data collection has been studied in prior work on incentivized exploration, for both public and private agent types. In particular, for public types one can collect samples for all type-arm pairs, under some assumptions \citep{ICexploration-ec15}. For arbitrary types, one can, in some sense, explore all type-arm pairs that can possibly be explored \citep{Jieming-multitypes18}. Recall that a sufficient amount of warm-up data needs to be guaranteed only w.r.t. agents' beliefs, in particular, only for stochastic adversaries.

\subsection{Proof of the positive result (\cref{thm:worst-UB-extension})}
\label{sec:hetero-UB-proof}

% \begin{lemma}[\citealp{freedman1975tail}]
% \label{lem:martingale uniform concentration}
% Given constant $c > 0, z<0$, let $Y_t$ be a martingale with respect to $X_t$ such that
% $|X_t| \leq c$ almost surely,
% for any $n\geq 1$, we have
% \begin{align*}
% \Pr\sbr{\max_{1\leq s \leq n} Y_s \leq z} \leq \exp\rbr{-\frac{z^2}{2(z+nc^2)}}
% \end{align*}
% \end{lemma}

First, we upper-bound the error probability for the MLE estimate, according to the agent's beliefs.
\begin{claim}\label[claim]{cl:hetero-UBpf-errorProb}
%\begin{align}\label{eq:hetero-UBpf-errorProb}
$\probE\sbr{\estimatestate\neq\truestate} \leq \eta/4$.
\end{claim}

The proof is a lengthy argument about log-likelihood ratios, deferred to \cref{apx:MLE}.

Next we show that the mechanism is BIR and BIC throughout the main stage.
Fix types
    $\type = \rbr{\pubtype,\subtype},\; \type' = \rbr{\pubtype,\reptype}\in \types$.
Using the notation from \cref{def:BIR} and \cref{def:BIC}, denote
\begin{align}\label{eq:hetero-UB-proof-Val}
\Val_{\mech}(\mystate,\type,\type')
:= \agentE\sbr{\Val_{\truestate}\rbr{a_t,\type}
    \mid \truE{t},\; \pubtype^{(t)} = \pubtype,\; \reptype^{(t)} = \reptype,\; \estimatestate=\mystate}
\end{align}
for each round $t$ in the main stage and each state $\mystate\in\support(\prior)$. In words, this is the Bayesian-expected utility of the agent type $\type$ if the reported type is $\reptype$ and the estimated state is $\estimatestate=\mystate$, assuming truthful reporting in the previous rounds. By specification of the mechanism, this quantity is the same for all rounds in the main stage. Write
    $\Val_{\mech}(\mystate,\type)= \Val_{\mech}(\mystate,\type,\type)$
for brevity.

To verify the BIR constraints, note that
the utility of agent type $\type$ from the mechanism is
\begin{align*}
\Val_{\mech}(\estimatestate,\type)
&\geq (1-\Pr\sbr{\estimatestate\neq\truestate}) \cdot
\Val_{\mech}(\truestate,\type)\\
&\geq  \Val_{\mech}(\truestate,\type) - \Pr\sbr{\estimatestate\neq\truestate}\\
&\geq \ValOut(\type).
\end{align*}
where the last inequality holds since $\Pr\sbr{\estimatestate\neq\truestate}\leq \frac{\eta}{4}$
and mechanism $\mech$ chooses an $\eta$-BIR
policy with probability $\frac{1}{2}$.

To verify the BIC constraints:
\begin{align*}
\Val_{\mech}(\estimatestate,\type)
&\geq \Val_{\mech}(\truestate,\type) - \Pr\sbr{\estimatestate\neq\truestate}\\
&\geq \, \Val_{\mech}(\truestate,\type,\type') + \frac{1}{2}\cdot\eta - \Pr\sbr{\estimatestate\neq\truestate} \\
&\geq \, \Val_{\mech}(\estimatestate,\type,\type') + \frac{1}{2}\cdot\eta - 2\Pr\sbr{\estimatestate\neq\truestate}\\
&\geq \Val_{\mech}(\estimatestate,\type,\type').
\end{align*}
% The first and the third inequality holds since the probability $\mystate\neq\estimatestate$ is at most $\Pr\sbr{\estimatestate\neq\truestate}$.
The second inequality holds since mechanism $\mech$ chooses an $\eta$-BIC
policy with probability $\frac{1}{2}$.

% Given that $\mech$ is BIC and BIR,
% \eqref{eq:lm:mech-single} follows immediately by invoking \eqref{eq:IPS-UB} and the fact that $T_0\leq T/2$
% and the sampling probability of each action is at least half of the that for type $\type_t$ in the benchmark.
%Given that $\mech$ is BIC and BIR,

Finally, we derive the error bound for statistical estimation. Invoking inequality \eqref{eq:IPS-UB},
we have
\begin{align*}
\err\rbr{\mech\mid\adv}
&\leq \max_{a\in[n]}\frac{1}{(T-T_0)^2}\sum_{t\in [T_0,T]} \frac{1}{\sampleD_t(a)} \\
&\leq \frac{1}{(T-T_0)^2}\;\max_{a\in[n]} \max_\mystate
\sum_{\type}\sum_{t\in [T_0,T]}
\frac{\ind{\type_t = \type}}{\frac{1}{2}\cdot\hatbenchpolicy_a(\mystate,\type)}\\
&\leq \frac{2\,\worstbench + 2\delta}{T-T_0}.
% \frac{1}{(T-T_0)^2}\;\max_{a\in[n]} \max_\mystate
% \sum_{\type}
% \rbr{\frac{(T-T_0)\cdot \typeFreqEst(\type)}{\frac{1}{2}\cdot\benchpolicy_a(\mystate,\type)}
% + \frac{(T-T_0)\cdot |\typeFreqEst(\type)-\typeFreq(\type)|}{\frac{1}{2}\cdot\min_{\type}\benchpolicy_a(\mystate,\type)}},
% %&\leq \frac{2}{T-T_0}
% %\rbr{\bench(\prior,\typeFreq) + \delta +  \frac{\eps}{\min_{\type,\mystate,a}\benchpolicy_a(\mystate,\type)}}.
\end{align*}

\section{Heterogeneous agents with estimated type frequencies}
\label{app:multi}

\newcommand{\diff}{\term{diff}}
    % the difference term in the theorem

This section focuses on type frequencies $\typeFreq$  in the main stage:
\begin{align}\label{eq:type-rates}
\typeFreq(\theta \mid \adv)
    = \frac{1}{T-T_0}\,\sum_{t>T_0}\;
    \ind{\type_t= \type}
    %\ind{(\pubtype^{(t)},\, \reptype^{(t)}) = \type}
        \qquad\forall\type \in\types.
\end{align}
The mechanism is initialized with estimated type frequencies $\typeFreqEst$: a distribution over types that estimates $\typeFreq(\cdot \mid \adv)$ for the (actual) adversary $\adv$. Such estimates can be constructed exogenously by a clinical trial (via surveys, medical history, or data from the past clinical trial) or an online platform (via user profiles and their interaction history).

The goal is to mitigate the influence of rare-but-difficult types. The benchmark $\worstbench$ in \cref{app:worst} (see \refeq{eq:worst-bench-alt}) is determined  by the ``bad" types (ones with large benchmark value) that are difficult for the mechanism to handle. To mitigate this when the bad types are rare, we define an alternative benchmark which replaces the $\max$ over all types with an expectation over $\typeFreq$:
\begin{align}\label{eq:bench-multi-generic}
\bench(\prior,\typeFreq)
     = \inf_{\text{BIC, BIR state-aware policies $\policy$}} \;\;
        \sup_{\mystate\in\support(\prior),\,a\in[n]} \;\;
%        \sum_{\text{types $\type$}}\frac{\typeFreqEst(\type)}{\policy_a(\type,\mystate)}.
          \E_{\text{type }\type\sim\typeFreq}
            \sbr{\frac{1}{\policy_a(\type,\mystate)}}.
\end{align}
This benchmark is no larger than the worst-case benchmark $\worstbench$, and can be much smaller; we provide a simple numerical example in \cref{apx:examples}.

We provide upper and lower bounds relative to $\bench(\prior,\typeFreq)$. Our mechanism  attains a strong guarantee when the estimated type frequencies  $\typeFreqEst$ are not too far from $\typeFreq$. When $\typeFreqEst=\typeFreq$, this guarantee is optimal, much like in \cref{thm:LB-single}, in the worst case over all adversaries with a given $\typeFreq$.

%We consider heterogeneous agents from a different perspective whereby we fix the empirical type frequencies $\typeFreq$ in the adversarially chosen type sequence. The benchmark $\worstbench$ in \cref{app:worst} takes a $\max$ over all types (see \refeq{eq:worst-bench-alt}), and may therefore be skewed towards ``bad" types -- ones with large benchmark value -- that are difficult for the mechanism to handle. To mitigate this issue when the bad types are rare, we define an alternative benchmark parameterized by $\typeFreq$, denoted $\bench(\prior,\typeFreq)$, which replaces the $\max$ over all types with an expectation over $\typeFreq$. Our mechanism attains guarantees with respect to this benchmark, when initialized with estimated type frequencies $\typeFreqEst$ that are not too far from $\typeFreq$. Such estimates can be constructed exogenously by a clinical trial (via surveys, medical history, or data from the past clinical trial) or an online platform (via user profiles and their interaction history). When $\typeFreqEst=\typeFreq$, this guarantee is optimal, much like in \cref{thm:LB-single}, in the worst case over all adversaries with a given $\typeFreq$.

%\begin{remark}\label{rem:new-bench}
%\end{remark}

The influence of rare-but-difficult agent types can now be understood solely from analyzing the benchmark $\bench(\prior,\typeFreq)$. It is easy to see that $\bench(\prior,\typeFreq)\leq \bench(\prior,\types)$.
In \cref{apx:examples} we provide a simple numerical example when $\bench(\prior,\typeFreq)$ provides a substantial improvement. However, we do not have
 $\bench(\prior,\typeFreq) = \E_{\type\sim\typeFreq}
        \bench(\prior,\{\type\})$,
and more generally $\bench(\prior,\typeFreq)$ does not appear to admit a clean characterization in terms of the per-type benchmarks. Instead, the influence of a particular type on the benchmark depends on the entire \emph{distribution} over types. The entire distribution could be ``bad", not just the individual types.

Thus, analyzing the influence of ``bad" types/distributions beyond numerical examples is quite subtle. In \cref{app:bad}, we consider a scenario with two type distributions, ``good" and ``bad", and their mixtures. We analyze the benchmark and show that the influence of the ``bad" type distribution on the mixture is indeed mitigated, in some asymptotic sense.

%We emphasize that we only consider empirical type sequences in the adversarially chosen type sequence. In particular, we do not require the agent types to be drawn i.i.d. in each round.

\subsection{Results}

While we add some new ideas, we reuse much of the machinery developed in \cref{app:worst}. For starters, we re-use the non-degeneracy condition \eqref{eq:worst-nondegen}, with the same proof as in \cref{cl:hetero-finite}.

\begin{claim}\label[claim]{cl:hetero-freq-finite}
Under \eqref{eq:worst-nondegen}, $\bench(\prior,\typeFreq)<\infty$.
\end{claim}

Let
    $\benchpolicy = \benchpolicy(\prior,\typeFreq)$
be the benchmark-optimizing policy, \ie a state-aware policy that optimizes the $\inf$ in \refeq{eq:bench-multi-generic}, ties broken arbitrarily. Such policy exists as an optimizer of a continuous function on a compact space.

We reuse mechanism
    $\mechname(T_0,\benchpolicy,\bestA[])$
from \cref{app:worst}
where $\bestA[]$ is the best-arm policy and
    $\benchpolicy = \benchpolicy(\prior,\typeFreqEst)$
is the benchmark-optimizing policy relative to the \emph{estimated} type frequencies. The number of warm-up samples is given by the same  \refeq{eq:worst-N}.

The performance guarantee becomes more complex, as it depends on $\typeFreq$ and $\typeFreqEst$.

\begin{theorem}\label{thm:hetero-UB}
\sloppy
Assume \eqref{eq:worst-nondegen}. Let $\typeFreqEst$ be some distribution over types.
Consider mechanism $\mechname(T_0,\benchpolicy,\bestA[])$, where policy
    $\benchpolicy = \benchpolicy(\typeFreqEst)$
optimizes the benchmark for distribution $\typeFreqEst$. This mechanism satisfies part (b) in \cref{thm:mech-hetero}
with parameters
    $\initProb = \eta^*/8$
and
    $N_{\prior} = N_{\prior}(\eta^*)$,
as per Eqns. \eqref{eq:worst-nondegen-gap} and \eqref{eq:worst-N}. The performance guarantee is as follows:
\begin{align}\label{eq:thm:hetero-UB}
\err\rbr{\mech\mid\adv}
\leq \frac{2\rbr{\bench(\prior,\typeFreq)
        + C_{\prior}(\typeFreq,\typeFreqEst) \cdot \norm{\typeFreqEst-\typeFreq}{1}}}
    {T-T_0},
\end{align}
where
    $\typeFreq(\cdot) = \typeFreq(\cdot\mid\adv)$
and the normalization factor $C_{\prior}(\typeFreq,\typeFreqEst)<\infty$ is determined by the prior $\prior$, the utility structure, and distributions $\typeFreq,\typeFreqEst$.

In more detail,
    $C_{\prior}(\typeFreq,\typeFreqEst)
        = C_{\prior}(\typeFreq) + C_{\prior}(\typeFreqEst)$,
    where
\begin{align*}
C_{\prior}(\typeFreq)
    := \rbr{\min_{\type\in\types,\;\mystate\in \supp(\prior),\; a\in[n]}\quad
                \benchpolicy_a(\mystate,\type)}^{-1},
\qquad \benchpolicy = \benchpolicy(\prior,\typeFreq).
\end{align*}
\end{theorem}

% \ascomment{revised the remark, again}

\begin{remark}
We can reasonably hope that the difference term in \cref{thm:hetero-UB},
    \[ \diff := C_{\prior}(\typeFreq,\typeFreqEst) \cdot \norm{\typeFreqEst-\typeFreq}{1},\]
vanishes as the time horizon $T\to\infty$. Indeed, suppose the realized type frequencies $\typeFreq$ are fixed, and the estimates $\typeFreqEst$ are based on the observations from the warm-up stage. The estimation error
    $\norm{\typeFreqEst-\typeFreq}{1}$
scales as $\tildeO\rbr{1/\sqrt{T_0}}$ with high probability under i.i.d. type arrivals, which is
    $\tildeO\rbr{T^{-1/4}}$
if the warm-up stage lasts for (say) $T_0 = \sqrt{T}$  rounds.%
\footnote{Here, we use $\tildeO(\cdot)$ notation to hide $O\rbr{\sqrt{\log T}}$ factors.}
Therefore, we can reasonably hope that the estimation error scales as some negative power of $T$ in practice: say, as at most $T^{-\gamma}$ for some $\gamma>0$. Then, fixing an arbitrary  $\delta\in(0,1)$, for any $T>\delta^{-1/\gamma}$ we have
\begin{align*}
\diff \leq
C_{\prior,\delta}(\typeFreq)\cdot T^{-\gamma},
\quad \text{where }
C_{\prior,\delta}(\typeFreq):=
    \sup_{\typeFreqEst:\;\norm{\typeFreqEst-\typeFreq}{1}<\delta }
        C_{\prior}(\typeFreq,\typeFreqEst).
\end{align*}
The point here is that $C_{\prior,\delta}(\typeFreq)<\infty$ does not depend on $\typeFreqEst$ or the time horizon $T$.
\end{remark}

\cref{thm:hetero-UB} can be extended to ``relaxed" versions of $\bestA[]$ and $\benchpolicy$, just like in \cref{thm:worst-UB-details}. The exact formulation is omitted, as we believe it does not yield additional insights.

\medskip

We prove that the guarantee in \refeq{eq:thm:hetero-UB} (with $\typeFreqEst=\typeFreq$) is essentially optimal in the worst case over all adversaries with given type frequencies $\typeFreq$.
%Essentially, the statement in \cref{thm:LB-single} holds for benchmark $\bench(\prior,\typeFreq)$.

\begin{theorem}\label{thm:worst-LB est}
Fix prior $\prior$ and time horizon
    $T\geq \Omega\rbr{\bench(\prior,\typeFreq)^{-2}}$
and $T_0\leq T/2$.
Let $\typeFreq$ be some distribution  over types such that
    $\typeFreq(\cdot)\in \tfrac{\Z}{T-T_0}$.
Let $S$ be a contiguous subset of rounds in the main stage, with $|S|\geq c\cdot T$ for some absolute constant $c>0$. Suppose mechanism \mech is non-data-adaptive over $S$.
Then some adversary $\adv$ satisfies
$\typeFreq(\cdot\mid\adv) = \typeFreq$ and
\begin{align}\label{eq:thm:LB-multi}
 \err\rbr{\mech\mid\adv} = \Omega\rbr{\frac{\bench(\prior,\typeFreq)}{T}}.
\end{align}
\end{theorem}

Like in the previous section, the lower bound is proved using the techniques from the homogeneous case (\cref{thm:LB-single}). We defer the proof to \cref{apx:LB}.

\subsection{Proof outline for the positive result (\cref{thm:hetero-UB})}

% \ascomment{Lots of rewrites in this subsection, could you pls re-read just in case?}

We first bound the difference in the benchmark due to the estimated type frequency.

\begin{lemma}\label[lemma]{lm:hetero-benchmarks}
Let $\typeFreq,\typeFreqEst$ be distributions over types. Then
\begin{align*}
\bench(\prior,\typeFreqEst) - \bench(\prior,\typeFreq)
\leq C_{\prior}(\typeFreq)\cdot \norm{\typeFreqEst-\typeFreq}{1}.
\end{align*}
\end{lemma}
\begin{proof}
Letting $\benchpolicy = \benchpolicy(\prior,\typeFreq)$, we have
\begin{align*}
\bench(\prior,\typeFreqEst)
&\leq \sup_{\mystate\in\support(\prior),\,a\in[n]} \;\;
\E_{\text{type }\type\sim\typeFreqEst}
\sbr{\frac{1}{\benchpolicy_a(\type,\mystate)}} \\
&\leq \sup_{\mystate\in\support(\prior),\,a\in[n]} \;\;
\E_{\text{type }\type\sim\typeFreq}
\sbr{\frac{1}{\benchpolicy_a(\type,\mystate)}}
+ C_{\prior}(\typeFreq)\cdot\norm{\typeFreqEst-\typeFreq}{1}\\
&= \bench(\prior,\typeFreq)
+ C_{\prior}(\typeFreq)\cdot \norm{\typeFreqEst-\typeFreq}{1}.\qedhere
\end{align*}
\end{proof}

The proof that the mechanism satisfies BIC and BIR is identical to that for \cref{thm:worst-UB-details}, essentially because these two properties do not depend on the type frequencies in the main stage. Hence the details are omitted here.

To derive the error bound, recall that by the specification of our mechanism, the sampling probabilities in the main stage can be lower-bounded as
    $\sampleD_t(a) \geq \tfrac12\; \benchpolicy_a\rbr{\estimatestate,\type_t}$
for each arm $a$ and round $t$, where the benchmark-optimizing policy is
    $\benchpolicy = \benchpolicy(\typeFreqEst)$.
Invoking the generic guarantee \eqref{eq:IPS-UB} for IPS estimators, conditional on the data in the warm-up stage we have
\begin{align*}
\err\rbr{\mech\mid\adv}
&\leq
    \frac{2}{(T-T_0)^2}
    \max_{\text{arms }a}\;
    \sum_{t\in [T_0,T]}\;
    \frac{1}{\benchpolicy_a\rbr{\estimatestate,\type_t}},
\\
&\leq
    \frac{2}{(T-T_0)^2}\;
    \max_{\text{arms } a}\;
    \max_{\mystate\in\supp(\prior)}\;
    \sum_{\type\in\types}\;
    \sum_{t\in [T_0,T]}\;
    \frac{\ind{\type_t = \type}}{\benchpolicy_a(\mystate,\type)}\\
&\leq
    \frac{2}{T-T_0}\;
    \max_{\text{arms }a}\;
    \max_{\mystate\in\supp(\prior)}\;
    \sum_{\type\in\types}\;
    \frac{\typeFreq(\type)}{\benchpolicy_a(\mystate,\type)} \\
&\leq
    \frac{2}{T-T_0}\;
    \max_{\text{arms }a}\;
    \max_{\mystate\in\supp(\prior)}\;
    \sum_{\type\in\types}\;
    \frac{\typeFreqEst(\type) + |\typeFreqEst(\type)-\typeFreq(\type)|}
         {\benchpolicy_a(\mystate,\type)} \\
&\leq
    \frac{2}{T-T_0}\;
    \rbr{ \bench(\prior, \typeFreqEst)
            + C_\prior(\typeFreqEst) \cdot \norm{\typeFreqEst-\typeFreq}{1}}.
\end{align*}
The final guarantee, \refeq{eq:thm:hetero-UB}, follows by plugging in \cref{lm:hetero-benchmarks}.

% which is upper-bounded by the right-hand side of \eqref{eq:thm:worst-UB}, as claimed.

\subsection{Vanishing impact of bad types}
\label{app:bad}

% distributions over types: resp., good, bad and mixed
\newcommand{\goodF}{F_{\mathtt{good}}}
\newcommand{\badF}{F_{\mathtt{bad}}}
\newcommand{\mixF}{F_{\eps}}

% good/bad policy
\newcommand{\goodS}{\policy^{\mathtt{good}}}
\newcommand{\badS}{\policy^{\mathtt{bad}}}

Let us analyze benchmark
    $\bench(\prior,\typeFreq)$
to clarify its key advantage: the ability to limit the influence of ``bad" types (or ``bad" distributions over types). This is an intricate point, as discussed early in this section. 

We focus on a simple scenario with two distributions over types, $\goodF$ and $\badF$. We posit that $\bench(\prior,\badF)$ is ``large" and $\bench(\prior,\badF)$ is ``small". The base case is that $\goodF,\badF$ are each supported on a single type: resp., a ``good type" and a ``bad type". More generally, $\goodF$ (resp., $\badF$) can be supported on a subset of ``good types" (resp., ``bad types"); further, we allow their support sets to  overlap on some ``medium types".
%For example, $\badF$ could be supported on a single ``bad type", or a subset of ``bad types" disjoint from the support of $\goodF$.

We consider mixtures of $\goodF$ and $\badF$: distributions
\begin{align}\label{eq:bad-mixtures}
\mixF \triangleq (1-\epsilon)\cdot \goodF + \eps\cdot\badF,
\quad \eps\in[0,1].
\end{align}
Interestingly, the bad influence of $\badF$ persists even in these mixtures:
\begin{align}\label{eq:bad-influence}
\bench(\prior,\mixF)
\geq (1-\eps)\cdot \bench(\prior,\goodF)
+ \eps\cdot \bench(\prior,\badF),
\qquad\forall \eps\in[0,1].
\end{align}

We prove that the influence of $\badF$ vanishes as $\eps\to 0$, in the sense that
\begin{align}\label{eq:vanish bad types}
\lim_{\eps\to 0}\bench(\prior,\mixF)
= \bench(\prior,\goodF).
\end{align}

\begin{theorem}
Let $\goodF,\badF$ be distributions over $\types$ with
    $\bench(\prior,\goodF)<\bench(\prior,\badF)<\infty$.
Assume that restricting the set of types to $\supp(\goodF)$, there exists a state-aware policy that is BIC and $\eta$-BIR, for some $\eta>0$. Then \refeq{eq:vanish bad types} holds.
\end{theorem}
\begin{proof}
The ``$\geq$" direction in \eqref{eq:vanish bad types} holds by taking $\eps\to 0$ in \eqref{eq:bad-influence}. We focus on the ``$\leq$" direction.

Let $\goodS$ and $\badS$ be the optimal state-aware policy given type frequency $\goodF$ and $\badF$ respectively, modified to offer the outside option for any type
    $\type\in \supp(\goodF)\cap \supp(\badF)$.
Formally, both $\goodS$ and $\badS$ return an arm
$a\in\outside$ which maximizes
    $\E_{\mystate\sim\prior}\sbr{\Val_{\mystate}\rbr{ a,\type }}$.
%That is, for policy $\goodS$, if there is a type $\type \in \supp(\badF)\backslash \supp(\goodF)$,
%then $\goodS$ offers the outside option for that type.
%$\badS$ is augmented analogously.
Let $\policy^{\eta}$ be the BIC and $\eta$-BIR policy given type frequency $\goodF$ for some fixed $\eta>0$, modified similarly. Let
\begin{align}
\pmin = \min_{\type\in \supp(\goodF),\;a\in[n],\;\mystate\in\supp(\prior)} \goodS_a(\mystate,\type).
\end{align}
Note that that $\pmin > 0$ since $\bench(\prior,\goodF)$ is finite.

Consider the policy $\policy^*$ that offers two options to the agent and the agent can choose the best option based on his type:
\begin{enumerate}[{Option} 1:]
    \item follow policy $\goodS$ with probability $1-\eps^{1/2}$, and follow policy $\policy^{\eta}$ with probability $\eps^{1/2}$.
    \item follow policy $\badS$ with probability $\eta\eps^{1/2}$, and always recommends the outside option with probability $1-\eta\eps^{1/2}$.
\end{enumerate}
For any type $\type\in\supp(\goodF)$, we have
\begin{align*}
\Val_{\prior}\rbr{(1-\eps^{1/2})\cdot\goodS + \eps^{1/2}\cdot\policy^{\eta}}
&\geq (1-\eps^{1/2})\ValOut + \eps^{1/2}(\ValOut + \eta)\\
&\geq \eta\eps^{1/2}\cdot\Val_{\prior}\rbr{\badS} +(1-\eta\eps^{1/2})\cdot\ValOut.
\end{align*}
Thus any type $\type\in\supp(\goodF)$ will prefer option 1 to option 2.
Moreover, given option $1$, type~$\type$ agent will not have incentives to deviate his report since both $\goodS$ and $\policy^{\eta}$ are BIC for the type space $\supp(\goodF)$,
and deviating the report to types in $\supp(\badF)\backslash\supp(\goodF)$
will lead to this agent receiving value $\ValOut$, which is not beneficial for the agent.

For any type $\type \in \supp(\badF)$,
either he will choose option 2,
in which case in state $\mystate$ each action $a$ is sampled with probability at least
$\eta\eps^{1/2}\cdot \badS_a(\mystate,\type)$,
or he will choose option 1 and deviate the report to another type in $\supp(\badF)$,
in which case each action $a$ is sampled with probability at least $\pmin$ given any state.

Thus, by applying policy $\policy^*$ to the benchmark problem with type frequency $\mixF$, we have
\begin{align*}
\bench(\prior,\mixF)
\leq \frac{1-\eps}{1-\eps^{1/2}}\cdot \bench(\prior,\goodF)
+ \eps\cdot \rbr{\frac{1}{(1-\eps^{1/2})\pmin}
+ \frac{1}{\eta\eps^{1/2}}\cdot\bench(\prior,\badF)}.
\end{align*}
We complete the proof by taking the limit $\eps\to 0$.
\end{proof}

\section{Conclusions and open questions}
\label{sec:conclusions}
We introduce a model for incentivized participation in RCTs. The model combines the statistical objective (which is standard for clinical trials and adversarial in nature) from the economic constraints (which are standard in economic theory and based on Bayesian-stochastic beliefs). The model emphasizes homogeneous agents as the paradigmatic case, and extends to heterogeneous agents. The latter are endowed with public types that affect objective outcomes, and private types that affect subjective preferences over these outcomes.

We consider three model variants: homogeneous agents, heterogeneous agents with no ancillary information, and heterogeneous agents with estimated type frequencies.
%\asdelete{The third variant mitigates the influence of rare-but-difficult agent types.}
For each model variant, we identify an abstract ``comparator benchmark", design a simple mechanism which attains this benchmark, and prove that no mechanism can achieve better statistical performance. Our mechanisms have a desirable property of not adapting to the data collected after the ``warm-up stage"; however, our negative results apply to a much wider class of mechanisms.

% Q0: improve LBs to adaptive mechs
% Q1: reduce duration of warm-up stage by combining it with main stage
%     (medium-T regime, beyond asymptotics)
% Q2: unknown/partually known type space and/or distribution over types,
%     esp. without stochastic type arrivals
% Q3: agents can't quantify their own risk preferences (= private types)

Our work sets the stage for further investigation.
%
%First, we conjecture that ``data-adaptive" mechanisms (\ie mechanisms that can adapt to observations in an unrestricted way) cannot improve our guarantees on the estimation error, in the sense that all negative results (\cref{thm:LB-single,thm:worst-LB,thm:worst-LB est}) should essentially carry over to data-adaptive mechanisms. In particular, for homogeneous agents the key technical issue is to extend the lower bound on estimators (\cref{lm:prelims-ips}) to allow for data-adaptivity.
%
%\asdelete{one could ask concrete questions that go slightly beyond our modelling approach. Can one match our guarantees for heterogeneous agents if the type space is not fully known, or the type frequencies are known very imprecisely? As the duration of the warm-up stage could be prohibitively large in practice, particularly in the presence of ``difficult types", can one reduce it and/or use some of its data for the counterfactual estimation?}
%
First, can we reduce the sufficient duration of the warm-up stage by integrating some of the data-collection therein into the main stage?
Second, while we mitigate the influence of rare-but-difficult agent types for the \emph{main stage} (via estimated type frequencies), can one do it for (the sufficient duration of) the \emph{warm-up stage}, too?
%That is, reduce its duration and/or reuse its data for the counterfactual estimates.
%
Third, can one match our result for estimated type frequencies if they not given in advance but are instead  learned on the fly?
%
%Data-adaptive mechanisms appear necessary to make progress on these directions.%
%\footnote{For practical reasons, these mechanisms may be restricted to change their exploration strategy only a few times.}
%
%\asedit{Second, do we need the agents to believe in a simpler world? Put differently, can one modify our framework so that the set of feasible adversaries coincides with the support of agent beliefs, and still achieve similar results? We fall short of this goal, since our lower bounds require adversaries that lie outside of the support. Even better, have the statistical model \emph{coincide} with the agent beliefs.}
%
Fourth, can one meaningfully relax our ``statistical model"? That is, constrain the adversarial outcomes while still retaining the need for independent per-round randomization. The goal would be to derive stronger guarantees, positive and/or negative.
%
%Fourth, can one strengthen our ``economic model" -- allow beliefs over (some) adversarially chosen outcome sequences --  and still derive meaningful guarantees?
%
Fifth, and more speculatively: how to think about incentivized RCTs if the agents cannot fully quantify their own risk preferences, and therefore cannot fully report them?

%We conjecture that the impossibility result for homogeneous agents extends to non-adaptive mechanisms;  On the other hand, one might want to use an adaptive mechanism to attain the benchmark from \cref{sec:multi} without knowing anything about the type frequencies; whether this is possible is unclear.

%As far as the number of warm-up samples ($N_\prior$) is concerned, the main take-away of our results is that $N_\prior<\infty$ does not depend on the time horizon $T$.  Optimizing this quantity is another concrete technical question that is left open. A good case to focus on would be $n=2$ arms and homogeneous agents. We identify the ``prior gap" as the key parameter, as per \cref{thm:mech-single-details}, and achieve inverse-square dependence thereon.

\bibliographystyle{apalike}
\bibliography{bib-abbrv,bib-slivkins,bib-bandits,bib-AGT,bib-medical,bib-ML,ref}

% \newpage
% % shows the appendices in ToC
% \addtocontents{toc}{\protect\setcounter{tocdepth}{2}}
% \tableofcontents

\newpage
\appendix

% shows the appendices in ToC
\renewcommand{\contentsname}{APPENDICES}
\addtocontents{toc}{\protect\setcounter{tocdepth}{2}}
\tableofcontents

\medskip\medskip\medskip

\section{Simple numerical examples}
\label[myAppendix]{apx:examples}
%We provide simple numerical examples to illustrate our results.

\xhdr{Homogeneous agents.} There are $n=2$ treatments (henceforth, ``arms"), the ``old" and the ``new". As per our model, agents believe each arm gives the same utility distribution for all agents, and the \emph{state} (which we identify with a mapping from arms to utility distributions) is initially drawn from some prior $\prior$ (shared by all agents). Specifically, the old arm gives mean utility 1 with probability $\nicefrac{2}{3}$ and utility 0 with probability $\nicefrac{1}{3}$. The new arm gives utility 2 and -1 with probability $\nicefrac{1}{2}$ each. The prior is independent across arms. Given the mean utilities (and the fact that the realized utilities lie in the $[0,1]$ interval), the noise shape can be arbitrary for this example.%
\footnote{\label{fn:examples-states}Thus, in the formalism of our model, we have four possible states, with mean utility pairs $(1,2)$, $(1,-1)$, $(0,2)$ and $(0,-1)$, and respective prior probabilities $\nicefrac{1}{3}$, $\nicefrac{1}{3}$, $\nicefrac{1}{6}$, $\nicefrac{1}{6}$.}
%a more risky option which provides a higher utility to the agent is proven to be effective.
Agents' alternative to participation is to use the old arm, whose prior-mean utility is $\nicefrac{2}{3}$; so, the ``outside option" is $\ValOut = \nicefrac{2}{3}$. Note that it is ex-ante preferred to the new arm (whose prior-mean utility is only $\nicefrac{1}{2}$), so the agents are incentivized against participation if the arm probabilities are data-independent.

% value of the agent is his value for the existing treatment, which equals $\nicefrac{2}{3}$.

%In this example, if the designer for the clinic trial only uses naive unbalanced randomized trials for recruiting the patients, the patients will never have incentive to participate in the trial since the expected utility of the trial is always lower than the existing treatment. However, as illustrated in our paper, by harnessing the power of information asymmetry, the designer can implement a two-stage trial that randomizes between two treatments with minimum sampling probability close to $q=\nicefrac{4}{9}$ (which is the probability of the inferior treatment based on the data in the warm-up stage) in the main stage. This gives a benchmark value of

Our participation-incentivizing RCT design from \cref{sec:single} works out as follows. Leveraging \cref{lm:greedy-single}, consider $\eps^*$, the largest $\eps$ such that $\greedy{\eps}{\truestate}$ policy is BIR. We claim $\eps^* = \nicefrac{8}{9}$. Indeed, note that $\eps^*$ is
is the unique $\eps$ such that the Bayesian-expected utility of $\greedy{\eps}{\truestate}$, call it $U(\eps)$, equals the outside option $\ValOut$. And one can check that $U(\nicefrac{8}{9}) = \nicefrac{2}{3} = \ValOut$.
(One way to do that is to sum over the four possible states, as per \cref{fn:examples-states}, and for each state calculate the expected utility of the policy; we omit the easy details.)

Thus, we obtain the benchmark value:
\[ \bench(\prior)= n/\eps^* = \nicefrac{9}{4}.\]

%with corresponding benchmark value
%$\bench(\prior)=\nicefrac{9}{4}$ by applying an $\greedy{\eps}{\mystate}$ policy with $\epsilon = \nicefrac{8}{9}$. To see why this is the case, note that \cref{lm:greedy-single} has shown that it is sufficient to only consider $\greedy{\eps}{\mystate}$ policies for the benchmark value. Moreover,
%the optimal benchmark value is attained when the exploration probability $\eps$ is chosen such that the agent's utility for participation equals his outside option. Therefore, letting $q=\frac{\eps}{n} = \frac{\eps}{2}$ be the min sampling probability of the $\greedy{\eps}{\mystate}$ policy,
%the utility for participating in the trial is
%\begin{align*}
%\nicefrac{1}{6}(-q + 2(1-q) + 2(1-2q) + 2(2-q)) = \nicefrac{2}{3}
%\end{align*}
%where the latter is the agent's outside option value.
%By solving the equation, we have that $\eps=2q=\frac{8}{9}$.

Now let us calculate $N_{\prior}$, the sufficient number of samples for the warm-up stage, as per \refeq{eq:prior-gap}. The expression in \refeq{eq:prior-gap} has an appealing ``interpretability", characterizing $N_{\prior}$ in terms of the prior gap  $\gap(\prior)$ and the benchmark $\bench(\prior)$. However, as stated in the Introduction, we suspect one can improve over \refeq{eq:prior-gap} with a more elaborate analysis.

The calculation is as follows. The prior-expected utility of the best arm is
    $\Val_{\prior}(\bestA[])
     %  \yledit{=\nicefrac{1}{6}(2(2+1)+(2+0))} =
    = \nicefrac{4}{3}$
(the calculation is similar to that for $U(\eps^*)$.)
So, the ``prior gap'' (see \refeq{eq:prior-gap}) is $\gap(\prior)=\nicefrac{4}{3}-\nicefrac{5}{12}=\nicefrac{11}{12}$.
Therefore, according to \cref{eq:homo-N}, we have $\alpha = \nicefrac{22}{27}$ and \[N_P<144.\]
In other words, a warm-up stage with $144$ independent samples of each arm suffices to guarantee incentive-compatibility for the main stage. Recall that these samples may also be obtained exogenously, \eg via paid volunteers.

% In the model with heterogeneous agents, the benchmark and the sampling probabilities need to account for the incentives of the agents for misreporting the types. Suppose there are two types,

% We show that for the worst case setting, the benchmark \eqref{eq:worst-bench} can be simplified.
% \begin{align}
% \worstbenchhat
% = \max_{\text{types}\type\in\types} \;
% \bench(\prior,\{\type\}).
% \end{align}
% \begin{lemma}
% $\worstbenchhat=\worstbench$ for all prior $\prior$ and type space $\types$.
% \end{lemma}
% \begin{proof}
% First, it is easy to observe that $\worstbenchhat\leq\worstbench$, since $\worstbenchhat$ only requires BIR without incentive constraints. Next, we show that for any BIR policy $\policy$ with worst-case benchmark value $\worstbenchhat$, we can design a BIR and BIC policy that matches the same worst-case benchmark value.

% In particular, consider another policy $\hat{\policy}$ where for any type $\type\in\types$, we offer the agent the utility maximizing option among $\{\policy(\type')\}_{\type'\in\types}$.
% It is easy to verify that policy $\hat{\policy}$ is BIC and BIR. Moreover, the min sampling probability given policy $\hat{\policy}$ is still
% $$\min_{\mystate\in\support(\prior),\,a\in[n]}\;\;
% \min_{\text{types }\type\in\types} \;
% \policy_a(\mystate,\type),$$
% the same as policy $\policy$.
% \end{proof}

\xhdr{Heterogeneous agents.}
The worst-case benchmark coincides with the benchmark value for the worst type, by \cref{eq:worst-bench-alt}. The frequency-based benchmark, $\bench(\prior,\typeFreq)$ from \cref{app:multi}, can provide a significant improvement.
%is more intricate, as per \cref{rem:new-bench}.
We illustrate this point with a simple example below.

%However, when the principal has estimates regarding the type frequencies, the value of the benchmark is more intricate. In particular, the value of the benchmark may not be simple averages of the benchmark values when we consider the types separately due to the incentive constraints.

There are $n=2$ treatments (arms: the "old" and the "new"), one public type, and two private types. There are $3$ possible outcomes: ``full recovery", ``side effect", and ``no effect". Agents' alternative to participation is to choose the old arm.

Since we have a unique public type, all agents believe that the outcome distribution for a given arm is the same for all agents, like in the homogeneous case. Further, they believe the state --- mapping from arms to outcome distributions -- is initially drawn from some prior $\prior$, shared by all agents.
%The purpose of the trial is to figure out which treatment provides a higher probability of full recovery but suffers from a higher probability of side effect.
The prior is as follows. For each arm, there are two possible outcome distributions. The risky distribution has full recovery with probability $\nicefrac{1}{2}$, side effect with probability $\nicefrac{1}{3}$, while ineffective with probability $\nicefrac{1}{6}$. The conservative distribution has full recovery with probability $\nicefrac{1}{3}$, side effect with probability $\nicefrac{1}{6}$, while ineffective with probability $\nicefrac{1}{2}$. There are two possible states:
the old arm leads to risky distribution and the new arm leads to conservative distribution (call it state $\mystate_0$), and vice versa (state $\mystate_1$). The prior $\prior$ assigns probability $\nicefrac{2}{3}$ to state $\mystate_0$, and the remaining probability to state $\mystate_1$. This completes the description of agents' beliefs.

%with probability $\nicefrac{2}{3}$, the old arm leads to risky distribution, and the new arm leads to conservative distribution. With the other probability $\nicefrac{1}{3}$, the state is that the old arm leads to conservative distribution, and the new arm leads to risky distribution.

The two private types are $\type_0$ and $\type_1$, appearing with equal frequency. The associated subjective utilities are as follows. Type $\type_0$ has utility 1 for full recovery and utility 0 otherwise. Type $\type_1$ has utility 1 for full recovery, utility 0 for no effect, and utility -10 for side effects.

If the two types were considered separately, the benchmark values are $3$ for type $\type_0$ and $2$ for type $\type_1$
(with minimum sampling probabilities, resp., $\nicefrac{1}{3}$ and $\nicefrac{1}{2}$). The computation is similar to that for the homogeneous example, we omit the details.

Let $\policy_i$ be the benchmark-optimizing BIR policy for type $\type_i$, $i\in \{0,1\}$. Let $\policy$ be the joint policy that maps type $\type_i$ to the resp. per-type policy $\policy_i$. This policy is BIR since $\policy_0$ and $\policy_1$ are BIR. To see that policy $\policy$ is BIC, note that for each state, the two types have different best arms, and so misreporting the type would only decrease the probability of choosing the optimal arm.
% type $\type_0$ and $\type_1$ have different
%preferences for the treatment effects, i.e., their best arms are different in all states,
%the policies that optimizes the benchmarks when treating them separately are incentive compatible when considering them together.
Therefore, $\policy$ is a feasible policy for the definition \eqref{eq:bench-multi-generic} of $\bench(\prior,\typeFreq)$, and so the benchmark value is at most what one obtains by plugging in this policy, namely
    $\bench(\prior,\typeFreq)\leq \nicefrac{5}{2}$.
\footnote{The $\nicefrac{5}{2}$ value is obtained by explicitly computing the per-type policies $\policy_0$, $\policy_1$ and then plugging $\policy$ into \eqref{eq:bench-multi-generic}.}
Thus, we obtain a significant improvement over the worst-case benchmark $\worstbench=3$.
\section{Proof of Claim \ref{cl:hetero-UBpf-errorProb}: a claim on MLE error probability}
\label[myAppendix]{apx:MLE}
Consider the log-likelihood ratio (LLR) between two states
    $\mystate,\mystate'\in\support(\prior)$.
Specifically, let
    $\mystate_{(t)} = \mystate\rbr{a_t,\pubtype^{(t)}}$
be the outcome distribution for state $\mystate$ in a given round $t$, and let
    $\mystate_{(t,\outcome)}$
be the probability of outcome $\outcome$ according to this distribution. The LLR is defined as
\begin{align*}%\label{eq:hetero-UBpf-LLR}
\llr_t\rbr{\mystate,\mystate'}
    &:= \log\rbr{\mystate_{(t,\,\outcome)}/\mystate'_{(t,\,\outcome)}},
    &\outcome = \outcome_t.\\
\llr_S\rbr{\mystate,\mystate'}
    &:= \textstyle \sum_{t\in S}
        \llr_t\rbr{\mystate,\mystate'}
    &\forall S\subset[T].
\end{align*}
The MLE estimator is correct, $\estimatestate = \truestate$, as long as
\begin{align}\label{eq:hetero-UBpf-relevance}
\llr_{[T_0]}\rbr{\truestate,\mystate'}>0
\quad\text{for any other state $\mystate'\in \support(\prior)\setminus\cbr{\truestate}$}.
\end{align}
In the rest of the proof, we show that
\begin{align}\label{eq:hetero-UBpf-claim}
\probE\sbr{\text{\refeq{eq:hetero-UBpf-relevance} holds}} \geq 1-\eta/4.
\end{align}

Fix two distinct states
    $\mystate,\mystate'\in\support(\prior)$.
For each round $t$, define
\begin{align}
X_t := \llr_t\rbr{\mystate,\mystate'} - \KL\rbr{ \mystate_{(t)},\; \mystate'_{(t)} }.
\end{align}
Focus on the subset of the warm-up stage when the outcome distributions $\mystate_{(t)}$,  $\mystate'_{(t)}$ are distinct:
\begin{align}\label{eq:hetero-UBpf-S}
    S := \cbr{ t\in [T_0]:\; \mystate_{(t)} \neq  \mystate'_{(t)}}.
\end{align}
We apply concentration to $X_S := \sum_{t\in S} X_t$.
We are interested in the low-probability deviation
\begin{align}\label{eq:hetero-UBpf-concentration-target}
X_S \leq - |S|\cdot \KL_{\min}(\mystate,\mystate')
\end{align}
when $\truestate = \mystate$ and $S$ is large enough. Specifically, we prove the following claim.

\begin{claim}\label[claim]{cl:hetero-UBpf-concentration}
Fix states
    $\mystate,\mystate'\in\support(\prior)$
and an arbitrary $s_0\in [T_0]$. Then
\begin{align}\label{eq:hetero-UBpf-concentration}
\probE\sbr{\text{$|S|\geq s_0$ implies \eqref{eq:hetero-UBpf-concentration-target}}
    \mid \truestate = \mystate}
\leq \tfrac{1}{2\lr_{\min}(\prior)}\cdot e^{-2(s_0-1) \cdot \lr_{\min}(\prior)}.
\end{align}
\end{claim}
Let us defer the proof of \cref{cl:hetero-UBpf-concentration}
% \eqref{eq:hetero-UBpf-concentration}
till later, and use it to complete the proof of \cref{cl:hetero-UBpf-errorProb}.
To emphasize the dependence on the two states, write
    $S = S\rbr{\mystate,\mystate'}$
and denote the event in \eqref{eq:hetero-UBpf-concentration-target} with
    $\mE\rbr{\mystate,\mystate'}$.
%\begin{align}\label{eq:hetero-UBpf-concentration-event}
%X_S\rbr{\truestate,\mystate'} \leq - |S\rbr{\truestate,\mystate'}|\cdot \KL_{\min}(\truestate,\mystate').
%\end{align}
Taking the union bound in \eqref{eq:hetero-UBpf-concentration} over all states
    $\mystate'$
and setting $s_0 = N_\prior(\eta)$, and taking the Bayesian expansion over events
    $\cbr{\truestate=\mystate}$, $\mystate\in\support(\prior)$,
we derive
\begin{align}\label{eq:hetero-UBpf-union}
\probE\sbr{
    \text{$|S\rbr{\truestate,\mystate'}|\geq N_\prior(\eta)$ implies $\mE\rbr{\truestate,\mystate'}$}
        \quad\forall \mystate'\in \support(\prior)\setminus\cbr{\truestate}}
\leq \frac{\eta}{8}.
\end{align}

By assumption, with probability at least $1-\eta/8$, the mechanism collects at least $N_\prior(\eta)$ samples for each (action, public type) pair. Denote this event by $\mE_0$.  For any $\mystate'\in \support(\prior)\setminus\cbr{\truestate}$, there exists at least one (action, public type) pair $(a,\pubtype)$ such that
    $\truestate(a,\pubtype) \neq \mystate'(a,\pubtype)$.
Therefore, under event $\mE_0$ we have
    $|S\rbr{\truestate,\mystate'}|\geq N_\prior(\eta)$.
It follows that with probability at least $1-\eta/4$,
event $\mE\rbr{\truestate,\mystate'}$ holds for all states $\mystate'\in \support(\prior)\setminus\cbr{\truestate}$.
This implies \refeq{eq:hetero-UBpf-claim}, and therefore
completes the proof of \cref{cl:hetero-UBpf-errorProb}.

% It remains to prove \refeq{eq:hetero-UBpf-concentration}.

\begin{proofof}[\cref{cl:hetero-UBpf-concentration}]
Consider the random subset $S$ of the warm-up stage defined by \eqref{eq:hetero-UBpf-S}. Let $t(j)$
be the $j$-th element of $S$, for all $j\in \sbr{|S|}$. Define a new sequence of random variables:
    $\hat{X}_{0}=0$ and
\begin{align*}%\label{eq:hetero-UBpf-Xhat}
\hat{X}_j := \ind{j\leq |S|}\cdot X_{t(j)},\qquad \forall j\in[T_0].
\end{align*}

%We define a new sequence of random variables
%    $\cbr{\hat{X}_{\tau}}_{\tau\in [T_0]}$
%and function $t(\tau)$
%such that
%\begin{align}\label{eq:hetero-UBpf-old1}
%\hat{X}_1 := X_1;
%\quad t(1) := 1
%\end{align}
%and for any $\tau\geq 2$,
%$t(\tau)$ is the minimum $t' \in [t(\tau-1)+1,T_0]$ such that
%$\KL\rbr{ \mystate_{(t')},\; \mystate'_{(t')} } > 0$
%if such $t'$ exists, and $t(\tau) = T_0+1$ otherwise.
%Let
%\begin{align}\label{eq:hetero-UBpf-old2}
%\hat{X}_{\tau} := \indicate{t(\tau) \leq T_0}\cdot X_{t(\tau)}.
%\end{align}

Fix $\tau\in [T_0]$. Since
    $\agentE\sbr{\hat{X}_{j} \mid \hat{X}_{j-1}} = 0$ for all $j\in [T_0]$,
the sum
    $\hat{X}_{[\tau]} := \sum_{j\leq \tau}\hat{X}_{j}$
is a martingale (in the probability space induced by the agents' beliefs). This martingale's increments are bounded as $ \hat{X}_{j}\leq \ll_{\max}(\mystate,\mystate')$.
Applying the standard Azuma-Hoeffding inequality with deviation term $- \tau \cdot \KL_{\min}(\mystate,\mystate')$,
we obtain
\begin{align*}
\probE\sbr{\hat{X}_{[\tau]} \leq - \tau \cdot \KL_{\min}(\mystate,\mystate')\mid \truestate = \mystate}
    &\leq \exp\rbr{-\frac{2\tau
        \cdot \KL^2_{\min}(\mystate,\mystate')}{\ll^2_{\max}(\mystate,\mystate')}} \\
    &\leq \exp\rbr{-2\tau \cdot \lr_{\min}(\prior)}.
\end{align*}
%For any $s_0\leq T_0$, by
Taking a union bound over all $\tau \in [s_0, T_0]$,
we have
\begin{align}
&\probE\sbr{\hat{X}_{[\tau]} \leq - \tau \cdot \KL_{\min}(\mystate,\mystate'),\;
\forall\tau \in [s_0, T_0]
\mid \truestate = \mystate} \label{eq:prob uniform bound states} \\
&\qquad\leq
    \sum_{\tau \in [s_0, T_0]} e^{-2\tau \cdot \lr_{\min}(\prior)}
\leq \frac{1}{2\lr_{\min}(\prior)}
        \cdot e^{-2(s_0-1)\cdot \lr_{\min}(\prior)}. \nonumber
\end{align}
% Note that by the construction of sequence $\hat{X}_{\tau}$,
% the random variables that are skipped compared to $X_t$ is the zero valued one.
% Thus, letting $\tau(t)$ be the inverse of function $t(\tau)$,
% we have $\hat{X}_{\tau(t)} = X_{[t]}$ for any $t\in[T_0]$.
We obtain \eqref{eq:hetero-UBpf-concentration-target} by taking $\tau = |S|$ in \eqref{eq:prob uniform bound states}.
%Therefore, for any set $S$,
%when $|S|\geq s_0$, we have
%\begin{align*}
%X_S = \hat{X}_{[|S|]} \leq - |S| \cdot \KL_{\min}(\mystate,\mystate')
%\end{align*}
%and the same probability bound from \eqref{eq:prob uniform bound states}
%applies for
%\eqref{eq:hetero-UBpf-concentration}.
\end{proofof}

\section{Lower bounds for heterogeneous agents}
%(\cref{thm:worst-LB,thm:worst-LB est})
\label[myAppendix]{apx:LB}

In this appendix, we prove the two lower bounds for heterogenous agents: \cref{thm:worst-LB} (for worst-case type frequencies) and \cref{thm:worst-LB} (for estimated type frequencies). Both theorems are proved via the techniques from the homogeneous case (\cref{thm:LB-single}).

We start with some common notation. Fix some mechanism $\mech$ which is non-data-adaptive over the main stage, and an associated warm-up stage duration $T_0$. Fix some subset $S$ of rounds in the main stage 
%    $S \subset \cbr{T_0+1 \LDOTS T}$,
such that $|S| \geq c\cdot T$ for some absolute constant $c>0$. Given some state $\mystate$, let $\adv_{\mystate}$ be the adversary restricted to the first $t_0 = \min(S)-1$ rounds that samples outcome $\outcome_{a,t}$ from distribution $\mystate(a)$,
independently for each arm $a$ and each round $t\leq t_0$.
Note that $\adv_{\mystate}$ and \mech jointly determine the sampling distributions $\sampleD_t$ for all rounds $t\in S$. Let $N_{\mystate}(a,\type)$ be the number of times arm $a$ is selected in the rounds $t\in S$ for type~$\type$
under adversary $\adv_{\mystate}$.

\xhdr{Proof of \cref{thm:worst-LB}.}
%Similar to the proof of \cref{thm:LB-single},
There is a state $\mystate_0\in\supp(\prior)$,
type $\type\in\types$
and action $a_0\in[n]$
such that
% if the type frequency in the main stage coincides with $\typeFreq$,
\begin{align*}
\frac{1}{\frac{1}{|S|}\cdot\E\sbr{N_{\mystate_0}(a_0,\type)}} \geq \bench(\prior, \types).
\end{align*}
The adversary $\adv$ is constructed as follows. For the first $t_0$ rounds, types are drawn from the prior and we use the adversary $\adv_{\mystate_0}$. For rounds $t>t_0$, only type $\type$ arrives, and we use adversary $\adv^\dag$ from \eqref{eq:IPS-LB-body}, determined by the expected sampling probabilities $\rbr{\E[\sampleD_t]:\, t\in S}$ (see \cref{lm:prelims-ips}).
%
%Types are drawn from the prior in the warm-up stage, while only type $\type$ will occur in the main stage. For the outcome distributions, we use the adversary $\adv_{\mystate_0}$ for the warm-up stage, followed by adversary $\adv_S\rbr{\E[\sampleD_t]:\, t\in S}$ for the main stage, as defined in \eqref{eq:IPS-LB}.
%
This reduces to the single type estimation problem,
and by applying the same argument as in the proof of \cref{thm:LB-single}, we have that
\begin{align*}
\err\rbr{\mech\mid\adv}
\geq
    \Omega\rbr{\frac{\bench(\prior,\types)}{T}}.
\end{align*}

\xhdr{Proof of \cref{thm:worst-LB est}.}
%Similar to the proof of \cref{thm:LB-single},
There exists a state $\mystate_0\in\supp(\prior)$
and action $a_0\in[n]$
such that
% if the type frequency in the main stage coincides with $\typeFreq$,
\begin{align*}
\sum_{\text{types }\type\in\types}
\frac{\typeFreq(\type)}{\frac{1}{|S|\cdot \typeFreq(\type)}\cdot\E\sbr{N_{\mystate_0}(a_0,\type)}} \geq \bench(\prior, \typeFreq).
\end{align*}
The adversary $\adv$ is constructed like in the homogeneous case: $\adv_{\mystate_0}$ for the first $t_0$ rounds, and $\adv^\dag$ from \eqref{eq:IPS-LB} afterwards. 
%We use the adversary $\adv_{\mystate_0}$ for the preliminary stage, followed by adversary $\adv_S\rbr{\E[\sampleD_t]:\, t\in S}$ for the main stage, as defined in \eqref{eq:IPS-LB}.
By the argument from the proof of \cref{thm:LB-single}, we have that
\begin{align*}
\err\rbr{\mech\mid\adv}
\geq
    \Omega\rbr{\frac{1}{|S|^2}\;
    \sum_{\text{types }\type\in\types}\rbr{\sqrt{|S|}\cdot |\loS^{\type}| + \frac{|\hiS^{\type}|^2}{\E\sbr{N_{\mystate_0}(a_0,\type)}}}}.
\end{align*}
Here, letting $q_t := \E\sbr{\sampleD_t(a_0)}$ and $\type\in\types$, we define
\begin{align*}
\loS^{\type} &= \cbr{t\in S:\; \type_t=\type \text{ and } q_t\leq 1/\sqrt{|S|}},\\
\hiS^{\type} &= \cbr{t\in S:\; \type_t=\type \text{ and } q_t> 1/\sqrt{|S|}}.
\end{align*}

If there exists some type $\type\in\types$ such that
$|\loS^{\type}|\geq \bench(\prior, \typeFreq)\cdot \sqrt{|S|}$,
then the desired guarantee \eqref{eq:thm:LB-multi} holds immediately.
Otherwise, since $|S|\geq \Omega\rbr{\rbr{\tfrac{\bench(\prior)}{\min_{\type}\typeFreq(\type)}}^2}$,
we have $|\hiS^{\type}| \geq \frac{1}{2} |S|\cdot \typeFreq(\type)$
for all type~$\type$
and \eqref{eq:thm:LB-multi} holds as well.

%\section{Statistical estimators}
%\label{apx:probability}
\section{Statistical lower bound: Proof of \refeq{eq:IPS-LB-body}}
\label[myAppendix]{app:stats-LB}

We prove the lower bound from \refeq{eq:IPS-LB-body}. Let's provide a standalone formulation for this result. Formally, we interpret $\estF$ as a mapping from arm $a$ and history $\history_{T+1}$ to $[0,1]$, and call any such mapping an \emph{estimator}.

\begin{lemma}\label[lemma]{lm:prelims-ips}
Fix any subset $S\subset [T]$. Consider a mechanism \mech which draws the tuple of sampling distributions
    $\rbr{\sampleD_t: t\in S}$
from some fixed distribution, and does so before round $\min(S)$. For any estimator $\estF$ in the mechanism there is an adversary \adv such that
\begin{align}\label{eq:IPS-LB}
\err\rbr{\mech\mid\adv}
    \geq %\textstyle
    \Omega\rbr{\frac{1}{T^2}\;\max_{a\in[n]}\;\sum_{t\in S} \min\cbr{\frac{1}{\E\sbr{\sampleD_t(a)}}, \sqrt{|S|}}}.
\end{align}
The adversary in \eqref{eq:IPS-LB} only depends on $\estF$ and the expected sampling probabilities $\E\sbr{\sampleD_t}$.
\end{lemma}

\begin{remark}
To compare to the IPS upper bound in \refeq{eq:IPS-UB}, consider the case when $S$ is the entire main stage and the sampling distributions
    $\rbr{\sampleD_t: t\in S}$
are chosen deterministically. We see that IPS is worst-case optimal as long as the sampling probabilities are larger than $1/\sqrt{|S|}$.

\cref{lm:prelims-ips} was known for the special case when $S = [T]$ and each distributions $\sampleD_t$ is chosen as an independent random draw from some fixed distribution over $\Delta_n$ \citep{DR-StatScience14}. The general case, with an arbitrary subset $S$ and a correlated choice of distributions $\rbr{\sampleD_t: t\in S}$ may be a tool of independent interest. However, we did not attempt to optimize the constants.
\end{remark}

\medskip

We prove \cref{lm:prelims-ips} in the remainder of this subsection. We use Pinsker's Lemma to bound the difference in event probabilities.

\begin{lemma}[\citealp{pinsker1964information}]\label[lemma]{lem:pinsker}
If $P$ and $Q$ are two probability distributions on a measurable space $(X, \Sigma)$, then for any measurable event $A \in \Sigma$, it holds that
\[
\left| P(A) - Q(A) \right| \leq \sqrt{\tfrac{1}{2}\; \mathrm{KL}(P \| Q)},
\]
where
$\mathrm{KL}(P \| Q) = \int_X \ln \rbr{\mathrm d P/\mathrm d Q} \mathrm d P$
is the Kullback--Leibler divergence.
\end{lemma}

% For the lower bound, f
We observe that the bound is monotone in $\sampleD_t(a)$ for any arm $a$ and any time $t\in S$.
This holds since the principal can always simulate the estimator with lower sampling probabilities by randomly ignore the observations.
Thus it is sufficient to prove that if $\E\sbr{\sampleD_t(a)}>1/\sqrt{|S|}$ for all $t\in S$ and all $a\in [n]$, then
there is an adversary \adv such that
\begin{align*}
\err\rbr{\mech\mid\adv}
    \geq \Omega\rbr{\tfrac{1}{T^2}\;\max_{a\in[n]}\;\sum_{t\in S} \frac{1}{\E\sbr{\sampleD_t(a)}}}
        \qquad\text{for any estimator $\estFreq$}.
\end{align*}

We divide the analysis into two cases.

\xhdr{Case 1:}
There exists an arm $a^*\in[n]$ such that
$\sum_{t\in S} \tfrac{1}{\E\sbr{\sampleD_t(a^*)}} \geq 710T$.
Consider the binary outcome space $\{\outcome_0,\outcome_1\}$
and the estimator function $f(\outcome_t) = \ind{\outcome_t = \outcome_0}$.
Consider two stochastic adversaries $\adv$ and $\adv'$.
For arm $a^*$, at each time $t\in S$,
outcome $\outcome_0$ is sampled with probability $\tfrac{1}{2}$ in $\adv$,
and is sampled with probability $\tfrac{1}{2} + \delta_t$ in $\adv'$
where $\delta_t\in[0,\tfrac{1}{2}]$ to be specified later.
Moreover, both $\adv$ and $\adv'$ have the same arbitrary sequence of outcomes for time periods not in~$S$.
It is easy to verify that
\begin{align*}
\E\sbr{\advFreq(a^*)} - \E\sbr{f_{\adv'}^{S}(a^*)}
= \frac{1}{T}\sum_{t\in S} \delta_t.
\end{align*}
Moreover, since at any time $t$, arm $a^*$ is chosen with probability $\sampleD_t(a^*)$,
the KL-divergence between $\adv$ and $\adv'$ is
\begin{align*}
\mathrm{KL}(\adv \| \adv')
&= \E\sbr{\sum_{t\in S} \sampleD_t(a^*) \mathrm{KL}(\adv_t \| \adv'_t)} \\
&= \E\sbr{\sum_{t\in S} \sampleD_t(a^*) \rbr{\frac{1}{2}+\delta_t} \log (1+2\delta_t) + \rbr{\frac{1}{2}-\delta_t} \log (1-2\delta_t)} \\
&\leq 4\sum_{t\in S} \E\sbr{\sampleD_t(a^*)} \delta_t^2.
\end{align*}
Letting $P$ be the probability that the absolute difference between the estimator to $\E\sbr{\advFreq(a^*)}$
is at most $\tfrac{1}{2T}\sum_{t\in S} \delta_t$ given adversary $\adv$
and $P'$ be the same probability given adversary $\adv'$.
By \cref{lem:pinsker}, we have that
\begin{align}\label{eq:pf-ipsLB-diff}
|P-P'| \leq \sqrt{\tfrac{1}{2}\;\mathrm{KL}(\adv \| \adv') }
\leq \textstyle \sqrt{2\;\sum_{t\in S}\; \E\sbr{\sampleD_t(a^*)}\cdot \delta_t^2}.
\end{align}
Let us specify $\delta_t$'s to ensure that $|P-P'| \leq \tfrac{1}{2}$:
\begin{align}
 \delta_t = \rbr{\E\sbr{\sampleD_t(a^*)}\cdot\sqrt{8\sum_{t\in S} 1/\E\sbr{\sampleD_t(a^*)}}}^{-1}.
\end{align}

Note that
    $\E\sbr{\sampleD_t(a^*)} \geq 1/\sqrt{|S|}$
(since we've restricted attention to the case when all sampling probabilities are at least $1/\sqrt{|S|}$).
Moreover, $\E\sbr{\sampleD_t(a^*)}\leq 1$ implies that
\[ \textstyle \sum_{t\in S}\, 1/\E\sbr{\sampleD_t(a^*)} \geq |S|\]
and therefore
\begin{align*}
\delta_t \leq \frac{1}{\E\sbr{\sampleD_t(a^*)}\cdot\sqrt{8|S|}}
\leq \frac{1}{2\sqrt{2}}
\end{align*}
which implies that the choice of $\delta_t$ is feasible.

Thus, there exists an adversary, say $\adv$,
such that with probability at least $\frac{1}{4}$,
the absolute difference between the estimator to $\E\sbr{\advFreq(a^*)}$ is at least $\tfrac{1}{2|S|}\sum_{t\in S} \delta_t$.
Moreover, by Hoeffding's inequality, we have
\begin{align*}
\Pr\sbr{|\advFreq(a^*) - \E\rbr{\advFreq(a^*)}| \geq \frac{1}{4T}\sum_{t\in S} \delta_t }
\leq 2\exp\rbr{-\frac{2\rbr{\frac{1}{4}\sum_{t\in S} \delta_t}^2}{T}} \leq \frac{1}{8}.
\end{align*}
where the last inequality holds since $\sum_{t\in S} \tfrac{1}{\E\sbr{\sampleD_t(a^*)}} \geq 710T$.
By union bound, with probability at least $\frac{1}{8}$,
the distance between the estimator and the truth is at least $\frac{1}{4T}\sum_{t\in S} \delta_t$
and hence the expected mean square error is
\begin{align*}
&\E\sbr{\err\rbr{\mech\mid\adv}}
\geq \frac{1}{8} \rbr{\frac{1}{4T}\sum_{t\in S} \delta_t}^2
= \Omega\rbr{\frac{1}{T^2}\;\max_{a\in[n]}\;
\sum_{t\in S} \frac{1}{\E\sbr{\sampleD_t(a)}}} .
\end{align*}

\xhdr{Case 2:}
$\sum_{t\in S} \tfrac{1}{\E\sbr{\sampleD_t(a)}} \leq 710T$
for all arms $a\in[n]$.
% In this case, there exists an arm $a^*\in[n]$ such that $\sum_{t\in S} \tfrac{1}{\sampleD_t(a)} \leq 2|S|$
% since at each time $t$, the sum of sampling probabilities is at most $1$ and there are at least $2$ arms.
Consider again the binary outcome space $\{\outcome_0,\outcome_1\}$
and the estimator function $f(\outcome_t) = \ind{\outcome_t = \outcome_0}$.
Consider the stochastic adversaries $\adv$
such that for each arm $a$, at each time $t$,
outcome $\outcome_0$ is sampled with probability $\tfrac{1}{2}$ in $\adv$.
Let $S_{a}$ be the number of time periods such that arm $a$ is not chosen.
Note that given realized sequence of observations,
% there exists an arm $a^*\in[n]$ such that the total number of observations for arm $a^*$ is at most
there exists an arm $a\in [n]$ such that
$|S_{a}|\leq \frac{T}{2}$.
Moreover, since the adversary is independent across different time periods,
the observations is uninformative about the realization of the unobserved outcomes.
Thus, for any estimator in the mechanism,
the estimation error for arm $a^*$
is at least the variance of the outcomes in unobserved time periods and
\begin{align*}
\E\sbr{\err\rbr{\mech\mid\adv}}
\geq \E\sbr{\max_{a\in[n]}\;\var\sbr{\advFreq[S_a](a)}}
\geq \frac{1}{8T}
= \Omega\rbr{\frac{1}{T^2}\;\max_{a\in[n]}\;
\sum_{t\in S} \frac{1}{\E\sbr{\sampleD_t(a)}}}.
\end{align*}
Combining the observations, there must exist a deterministic sequence of adversary such that the mean square error conditional on the adversary satisfies stated lower bound.

\section{Upper bound for IPS: proof of \refeq{eq:IPS-UB}}
\label[myAppendix]{app:IPS-UB}

We note that it is not essential to focus on a particular subset $S$ of rounds such as the main stage; instead, the result holds for an arbitrary subset $S\subset [T]$.

Fixing any oblivious adversary, since the IPS estimator is unbiased, the mean square error of the IPS estimator equals it variance. Moreover, in the IPS estimator,
the random variable $\frac{\ind{a_t = a} \cdot f(\outcome_{a,t})}{\sampleD_t(a)}$
is independent across different time periods $t$.
Since the variance for the sum of independent random variables equals the sum of variance,
we have
\begin{align*}
&\err\rbr{\mech\mid\adv}\\
=& \max_{\outcome\in\outcomes,\,a\in[n]}\frac{1}{|S|^2}\sum_{t\in S} \var\sbr{\frac{\ind{a_t = a} \cdot f(\outcome_{a,t})}{\sampleD_t(a)}}\\
=& \max_{\outcome\in\outcomes,\,a\in[n]}\frac{1}{|S|^2}\sum_{t\in S} \rbr{\frac{1}{\sampleD_t(a)}-1} \cdot f(\outcome_{a,t})\\
\leq& \max_{a\in[n]}\frac{1}{|S|^2}\sum_{t\in S} \frac{1}{\sampleD_t(a)}.
\end{align*}

\end{document}